\documentclass[11pt]{amsart}
\usepackage{graphicx}
\usepackage{hyperref}
\usepackage{natbib}
\usepackage{color}

\newtheorem{thm}{Theorem}[section]

\newtheorem{lem}[thm]{Lemma}
\newtheorem{prop}[thm]{Proposition}
\theoremstyle{definition}
\newtheorem{defn}[thm]{Definition}

\theoremstyle{remark}
\newtheorem{rem}[thm]{Remark}

\numberwithin{equation}{section}

\newcommand{\set}[1]{\left\{#1\right\}}
\newcommand{\Real}{\mathbb R}
\newcommand{\Natural}{\mathbb N}

\newcommand{\such}{\ | \ }
\newcommand{\prob}{\mathbb{P}}

\newcommand{\expec}{\mathbb{E}}

\newcommand{\cov}{\mathbb{C} \mathsf{ov}}

\newcommand{\ud}{\mathrm d}
\newcommand{\inner}[2]{\left \langle #1 , #2 \right \rangle}

\newcommand{\hil}{\mathcal{X}}
\newcommand{\hili}{{\hil_i}}
\newcommand{\hilsub}{{\mathcal{Y}}}

\newcommand{\bB}{\mathbf{B}}

\newcommand{\bF}{\mathbf{F}}
\newcommand{\bH}{\mathbf{H}}
\newcommand{\bW}{\mathbf{W}}
\newcommand{\bX}{\mathbf{X}}
\newcommand{\bY}{\mathbf{Y}}
\newcommand{\bZ}{\mathbf{Z}}

\newcommand{\bbL}{\mathbb{L}}
\newcommand{\bbU}{\mathbb{U}}

\newcommand{\id}{\mathsf{id}}

\newcommand{\defp}{\mathbb{S}_{\succeq}}
\newcommand{\defpp}{\mathbb{S}_{\succ}}

\newcommand{\defppi}{\defpp^{\hil_i}}

\newcommand{\pdefp}{\prod_{\iii} \defp}
\newcommand{\pdefppi}{\prod_{\iii} \defppi}

\newcommand{\diag}{\mathsf{diag}}

\newcommand{\argmax}{\operatorname{argmax}}

\newcommand{\limn}{\lim_{n \to\infty}}

\newcommand{\pare}[1]{\left(#1\right)}
\newcommand{\bra}[1]{\left[#1\right]}
\newcommand{\dbra}[1]{[\kern-0.15em[ #1 ]\kern-0.15em]}
\newcommand{\dbraco}[1]{[\kern-0.15em[ #1 [\kern-0.15em[}
\newcommand{\dbraoc}[1]{]\kern-0.15em] #1 ]\kern-0.15em]}
\newcommand{\X}{\mathcal{X}}

\newcommand{\dfn}{\mathrel{\mathop:}=}

\newcommand{\nin}{n\in\Natural}

\newcommand{\iii}{i\in I}

\newcommand{\proj}{\pi}
\newcommand{\proji}{\proj_i}
\newcommand{\projj}{\proj_j}
\newcommand{\cop}{p^{\text{co}}}
\newcommand{\coq}{q^{\text{co}}}

\newcommand{\hilsubi}{{\mathcal{Y}_i}}
\newcommand{\defppsub}{\defpp^{\hilsub}}
\newcommand{\defppsubi}{\defpp^{\hilsub_i}}
\newcommand{\pdefppsubi}{\prod_{\iii}\defppsubi}

\begin{document}

\title[Price Impact Under Heterogeneous Beliefs and Restricted Participation]{Price Impact Under Heterogeneous Beliefs and Restricted Participation}
\author{Michail Anthropelos}
\address{Michail Anthropelos, Department of Banking and Financial Management, University of Piraeus.}
\email{anthropel@unipi.gr}

\author{Constantinos Kardaras}
\address{Constantinos Kardaras, Statistics Department, London School of Economics.}
\email{k.kardaras@lse.ac.uk}

\thanks{The authors acknowledge helpful feedback from Andrea Buffa, Scott Robertson, Jean-Charles Rochet, Marzena Rostek, Paul Schneider and Xavier Vives. 
}

\date{\today}%
\begin{abstract}
We consider a financial market in which traders potentially face restrictions in trading some of the available securities. Traders are heterogeneous with respect to their beliefs and risk profiles, and the market is assumed thin: traders strategically trade against their price impacts. We prove existence and uniqueness of a corresponding equilibrium, and provide an efficient algorithm to numerically obtain the equilibrium prices and allocations given market's inputs. We find that restrictions may increase the market's welfare if traders have different views regarding the covariance matrix of securities returns. The latter heterogeneity regarding covariance matrix disagreement is essential in modelling; for instance, when traders agree on the covariance matrix, restricting participation in some securities for some traders leaves equilibrium prices unaltered in the unrestricted securities, a certainly undesirable model effect.
\end{abstract}

\maketitle


\section*{Introduction}\label{sec:intro}

\subsection*{Discussion}

Traders often face restrictions on investing in certain type of financial assets. Several exogenous factors may prevent retail or institutional investors from accessing certain classes of securities. Standard examples are the proprietary trading of banks, allowable investments of mutual funds and pension funds, which by regulation are not allowed to hold certain instruments (e.g., non-investment graded bonds, classes of over-the-counter derivatives, securities in private placements, etc.).\footnote{According to Volcker rule (see e.g.~\cite{Whi11}), US banks are prohibited from trading several securities, derivatives and commodities for their own profit. Similar restrictions are imposed to European banks by Markets by the Financial Instruments Directive II---see e.g.~\cite{Bus17}).

Restrictions on the asset classes that regulated mutual funds are allowed to trade stem from two reasons. First, restrictions arise from the category that the mutual fund belongs; for example, a money market mutual fund is not allowed to invest in long-term bonds or equity, and a government bond mutual fund is not allowed to invest in securities issued by corporations. Second, registered mutual funds are generally not allowed to trade asset classes like private equity and commodities, which can be traded by other institutional investors; see~\cite{SEC16}, as well as~\cite{FulHon21}.

For regulatory investment restrictions for pension funds, we refer to the survey~\cite{OECD19}.}
Even for hedge funds, restrictions may stem from inflexible investment statement policies, whereby managers choose not to trade in some securities in order to emphasise the speciality of their investment strategies; see, among others,~\cite{Bla13}. 

Institutional investors may additionally avoid certain securities due to asymmetrically high transaction costs and margin requirements~\citep{KoiYog19}, or due to the difficulty in processing information related to assets' valuations---see~\cite{CorGop10} and the references therein. Such exogenous restrictions imply that some traders (such as mutual funds) cannot sufficiently diversify or hedge their risks, while their counterparties (such as hedge or private funds) can; a situation that affects the market's competitiveness and hence its risk-sharing efficiency. 

When several institutional traders are not allowed to trade some securities, large unrestricted traders gain significant market power to impact allocations and prices. In the financial literature, markers dominated by few large traders are usually called \emph{thin}~\citep{Rostek2016}. Such markets become non-competitive solely due to a small number of participants possessing market power. Market's non-competitiveness does not stem from asymmetric information or asymmetric exogenously-imposed bargaining power; rather, the structure of the security market is oligopolistic, where all traders can buy or sell the tradeable assets under a uniform-price double auction setting.

In markets with restricted participation, the assumption that traders are price-takers becomes particularly problematic, and ignoring traders' price impact is not consistent with observable practice. Furthermore, under participation restrictions, price impact has an additional component: a trader's strategy should take into account not only their, but also their counterparties' restrictions, and traders' actions also impact securities they are restricted from trading. In other words, traders recognise the impact they have on all equilibrium prices, thereby acting strategically, resulting in price impacts endogenously derived in equilibrium.

Even though institutional investors, whose participation is restricted, are some of the largest\footnote{In~\cite{BreSchSenSha20}, it is estimated that roughly 50\% of corporate bonds hold by institutional investors, while a recent study by~\cite{PI17} estimated that 80\% of equity markets is owned by institutional investors---see also the related discussions in~\cite{AskFarLji14} and~\cite{KoiYog19}.} protagonists of the financial markets, theoretical studies of equilibria with asymmetric restricted participation in markets with few large participants are rare. \emph{Our initial goal is to consider a non-competitive equilibrium model} (as the one in~\cite{RosWer15} and~\cite{MalRos17}) \emph{under a general structure of traders' restricted participation and prove the existence and uniqueness of such equilibrium}.

Besides asymmetric participation, traders' heterogeneity substantially impacts strategic behaviour. Recent theoretical models on price impact, such as~\cite{MalRos17},~\cite{Anth17},~\cite{BabPart19} and~\cite{AnthrKardVich20}, predict that traders' heterogeneity plays a crucial role, not only on the distribution of trading gains, but also on the market's total welfare. In the view of these insights, we choose to consider a model that is in line with the classic CARA-normal setting, with traders being heterogeneous on all of their characteristics, i.e., initial risky endowments, risk aversions and beliefs. 

Traders' deviations on risk aversion, initial endowment and payoff's expectations have been studied in both competitive and non-competitive security market models~\citep{RosYoo20}; however, herein we additionally consider traders with different views on the covariance matrix of the tradeable assets. This adds an extra non-trivial layer of analysis with respect to related literature, where the effect of different second moments on non-competitive equilibrium transactions has not been addressed.

It is rather unrealistic to assume that traders agree on the (co)variances of tradeable assets---see also the related argument in~\cite{Duc10}.\footnote{Difference on beliefs about (co)variances could also be supported by the large volume on derivatives whose underlying asset is realised variance, such as variance swaps. Although some traders of such assets are simply hedgers, a large number of them trade these products based on personal prediction or estimation regarding volatility and correlations of the associated securities---see~\cite{Bak15}.} In this manuscript, we aim to extensively study \emph{how this demonstrated traders' general disagreement on the second moments effects equilibrium prices, allocations, impacts and welfare}. Such additional heterogeneity on trader's beliefs is not included only for the sake of theoretical generalisation. It turns out that  deviation on beliefs regarding second moments of payoffs is an important factor that affects both equilibrium price impacts, as well as the induced market's efficiency. In particular, an indicative example demonstrates that if traders disagree simply on the correlation of payoffs, then participation's restrictions may increase welfare.\footnote{Under competitive market models, the effect of traders' disagreement about the asset's second moments on the equilibrium has already been highlighted. For example,~\cite{Duc10} show that disagreement about (co)variances heavily affects prices, while~\cite{Bak15} connect heterogeneity in beliefs with trading of contingent claims written on volatility.}

Recapitulating, we consider a model that includes three quite common market's features, all heavily affecting the formulation of the equilibrium transaction: (i) possible restricted participation of some traders for some securities; (ii) market's thinness, stemming solely from oligopolistic structure; and (iii) general traders' heterogeneity and especially to regards on their beliefs on the securities' second moments. To the best of our knowledge, this is the first work that simultaneously includes the combination of these three features.

\subsection*{Contributions}
 
We consider a static model and a finite number of traders under the standard CARA-normal setting as, for instance, in~\cite{Vay99} and~\cite{Viv11}, and study equilibrium pricing and allocations of securities, where traders can take both long and short positions. Under potential restricted participation, motive for trading stems from the heterogeneity in the traders' risky existing positions, risk profiles and beliefs.  

We deviate from a price-taking framework; equilibrium forms within a game played among traders, with all of them realising that their actions impact prices. We model market's operation as a uniform-price double auction and traders' strategic sets are the downward-sloping demand schedule they submit in the transaction, following the tradition of~\cite{Kyl89}. 

Under the aforementioned setting, Theorem~\ref{thm:nash} establishes existence and global uniqueness of a non-competitive Nash equilibrium. (The only additional imposed assumption is that at least three traders participate in the transaction of each security. The latter is necessary for equilibrium to exist; see~\cite{Kyl89}.)
To our knowledge, the present work is the first instance where existence and uniqueness of non-competitive equilibrium is shown under restricted participation and such extensive traders' heterogeneity. 

An algorithm that numerically calculates the equilibrium quantities is also provided. Besides its obvious computational value, this iterative procedure is economically motivated too, as it can be seen as a mechanism where traders update their best responses in a Walrasian type of auction. This procedure converges to the associate fixed point, as the Walrasian auction converges to its equilibrium.

Traders' disagreement on the covariance matrix brings interesting consequences on equilibrium transactions, that other forms of heterogeneity cannot yield. While both the covariance matrix and risk aversion are part of the risky component of traders' demand function, there are important effects in equilibrium that only heterogeneity on the covariance matrix unveils. In Theorem~\ref{thm:equal_prices}, we show that when traders agree on the covariance matrix, the equilibrium price of the unrestricted assets remains unchanged when restrictions to other assets are imposed. This is an unrealistic feature when assets are correlated and traders have price impact; traders' strategies, in principle, should take into account the pool of the assets that they are allowed to trade and hence restrictions should affect the equilibrium of all the market. This unreasonable result does not hold when traders disagree on the covariance matrix. Under such disagreement, we further show that equilibrium price impacts are \emph{not} necessarily monotone (in positive-semi-definite order) with respect to traders' covariance matrices, in contrast to the case of common trader beliefs, where there is monotonicity between trader's risk tolerance and price impact~\citep[Theorem 2]{MalRos17}. Lower estimated variances \emph{do not} necessarily imply higher price impact.

Of further importance is the effect of the discrepancy of covariance matrices on the market's welfare. We argue that lifting participation restrictions in a non-competitive market may \emph{not} be socially beneficial, via an indicative example where traders disagree on the covariance of the security payoffs. This comes in sharp contrast with no-price-impact competitive equilibrium, where full participation leads to Pareto optimal allocations and is always optimal, regardless of traders' heterogeneity. Boosts of efficiency through restrictions do  not necessarily stem from heterogeneity in traders' optimism; in our example, higher efficiency in restricted participation occurs when traders agree on expectations and variances of two securities, and only disagree on the correlation between them. In fact, this disagreement may result in reduction of the price impact of the trader who mostly benefits by trading. In such cases, it is possible that the total welfare decreases due to withdrawal of restrictions for that particular trader. 

It should be emphasised that different estimations on covariance matrices of tradeable assets are obviously market-specific, while traders' risk aversions (common or not) reflect trader's risk preferences in \emph{any} market. This is an important discrimination, since predictions of a model that are based on risk aversion heterogeneity in one market may be heavily violated to another when therein traders disagree on second moments. A simple and direct explanation, beyond the previous discussion, is that difference on covariances could be another source of mutually beneficial trading (for traders elsewhere homogeneous), since traders may estimate the correlation of their endowment with the tradeable asset differently.

\subsection*{Connections with existing literature}

Our work mostly relates to two strands of literature on security equilibrium pricing. On the one hand, we contribute to the ongoing research on imperfectly competitive financial markets, where traders are assumed strategic and heterogeneous; on the other hand, our results belong to the study of equilibrium models under restricted (also referred as limited or constrained) traders' participation.

Motivated by related empirical evidence, several authors have studied so-called \emph{thin} financial market models, where investors impact prices with their demands/orders; a recent and extent literature review on this strand is provided by~\cite{RosYoo20}. For the formulation of our non-competitive equilibrium model, we follow the uniform-price double-auction setting developed by~\cite{RosWer15}, where, in the spirit of~\cite{Kyl89} and~\cite{Vay99} and~\cite{Viv08}, traders act strategically in a demand submission game, trading against residual supply.~\cite{RosWer15} treat a dynamic version of this game, assuming however that traders have the same beliefs and risk aversion. The same Nash equilibrium setting is considered in~\cite{MalRos17}, where the focus is on the effect of such strategic behaviour when traders with potentially different risk aversions trade securities in decentralised markets.\footnote{Demand submission games in fragmented markets operating under double auctions are also studied in~\cite{CheDuff21},~\cite{RosYoo21a} and~\cite{Witt21}; however, therein demand schedules are not contingent on prices.}~\cite{RosYoo23} consider similar non-competitive equilibrium under incomplete demand conditioning. In all these models, traders are assumed to have same estimations for the (co)variances of the tradeable securities. In this paper, we highlight the importance that heterogeneity in securities' covariance matrices has on equilibrium price impacts and gains from trading, showing that it is not just a theoretical generalisation, and emphasising that the induced effects cannot be captured by heterogeneity on risk aversion.

There is also large literature on thin financial markets with exogenously given price-impact functions for each strategic trader---e.g.,~\cite{AlmCh00},~\cite{AlmThu05} and~\cite{HubSta04}. We deviate from this literature; similar to~\cite{Kyl89},~\cite{RosWer15} and~\cite{Viv11}, price impact in our model is derived endogenously as part of the market's equilibrium.

Additionally to heterogeneity of beliefs, and instead of decentralised markets or incomplete demand, we focus on \emph{restricted participation}. In our setting, we establish existence and \emph{global} uniqueness of Nash equilibrium; in contrast, uniqueness in the decentralised market setting has only been shown to hold locally. 

Restricted access may be regarded as a constraint on the traders' set of admissible 
portfolios. Among others, equilibrium models of financial markets that impose constrained sets have been studied by~\cite{Duff87} and~\cite{PolSic97}, by~\cite{BasCuo98} in a continuous-time setting, and more recently by~\cite{Cass06},~\cite{HenHerPre06},~\cite{AouCor09},~\cite{CorGop10}. All these papers consider competitive markets\footnote{Under competitive market structure, participation restrictions may arise endogenously---see, among others,~\cite{CarGorVil09} and~\cite{CavGonSod04}.},  where traders are essentially price-takers; price impact during transactions and its implication on market's efficiency are not addressed. However, several empirical studies---for instance,~\cite{KoiYog19},~\cite{HuHofLanTim19},~\cite{AllMatMat17},~\cite{FraIsrMosk18} and the related discussions in~\cite{NeuSoc20} and~\cite{RosYoo20}---have shown that large institutional investors comprise a significant part of the market's volume, with orders influencing security prices and eventually the portfolio allocations of all traders. Under a restricted participation setting, price impact has an additional dimension, since traders, when acting strategically, influence even transactions of assets they do not trade. 

Further to including price impact within a restricted participation setting, we show how different beliefs on payoffs' covariance matrix heavily affect the consequences of such restrictions to prices and the market's efficiency. In particular, and in contrast to the comparison between decentralised and centralised markets, restricted participation may result in higher efficiency even if traders have the same risk tolerance.\footnote{It is shown in~\cite{MalRos17} that centralised markets are always socially better than decentralised ones, when traders have the same risk aversion. We show here that this is not necessarily true when traders disagree on assets' second moments. The importance of strategic investors' heterogeneous beliefs is also highlighted in~\cite{BabPart19}. Therein, trading is done through intermediaries, and it is shown that when investors' disagreement is low, a fragmented market structure may arise endogenously at equilibrium; however, in centralised markets investors' welfare is always higher than in fragmented settings. Cases where fragmentation structure is beneficial occur under the equilibrium model of~\cite{CheDuff21} and~\cite{Witt21}, where however traders have common beliefs and traders' demands are not price contingent.} From this perspective, our work is also connected with the broader literature on financial market design, and especially on the effects of imposing or withdrawing participation's restrictions to certain market participants. Since equilibrium existence and uniqueness under any arbitrary restriction setting is established, we argue that difference of beliefs on second moments is a crucial factor that should be taken into account when analysing market's design, and especially participation rules.

Finally, there is a further link of our model with literature that considers asymmetric market power and segmented markets. For example,~\cite{TuckVi92},~\cite{RahZig09},~\cite{Zig04, Zig06} consider financial markets with restricted participation where market participants are distinguished to arbitrageurs and competitive investors; arbitrageurs have access to all tradeable assets and act strategically in a Cournot-type of framework, while investors are assumed price-takers. In our model, \emph{all} traders act strategically, making it more appropriate when large investors know that they can influence the market, even if they are restricted to trade only a subset of the securities.

\subsection*{Structure of the paper}

Section~\ref{sec:main} introduces the price-impact model and states the existence and uniqueness result. Section~\ref{sec:exampleSocialIneff} is dedicated to the importance of the beliefs' heterogeneity on the price impact and welfare issues. We conclude at Section~\ref{sec:conclu}, whereas all proofs are provided in Appendix~\ref{appsec:proofs}.

\section{Equilibrium Price Impact with Restricted Participation}\label{sec:main}

\subsection{Traders, securities and notation}
In the market, we consider a finite number of traders and use the index set $I$ to denote them. There is a finite number of tradeable risky securities, and their index set is denoted by $K$. Traders \emph{may} be restricted to trade some of the risky assets: this is modelled by assuming that trader $\iii$ has access to (effectively, is allowed to trade in) only a subset $K_i \subseteq K$ of the securities. In other words, trader $\iii$ may select units of securities in the subspace of $\hil \equiv \Real^K$ defined via
\[
\hili \dfn \set{x\in\hil \such x_j = 0, \ \forall \, j\in K \setminus K_i}, \quad \iii.
\]
We shall call \textbf{full participation} the market setting for which $K_i=K$, for all $\iii$. 
Before giving more details of the model's structure, we need to establish some necessary definitions and notation.  
For each $\iii$, we shall denote by $\proji$ the projection operator from $\hil$ on the space $\hili$: for $x\in\hil$, $\proji x$ has the effect of keeping all coordinate entries of $x$ corresponding to $K_i$ intact, while setting all coordinate entries of $x$ corresponding to $K \setminus K_i$ equal to zero. We also define $\defp$ as the set of all linear, symmetric, nonnegative-definite forms on $\hil$. On $\defp$, we define the partial order $\preceq$ via
\[
A \preceq B \quad \Longleftrightarrow \quad (B - A)\in\defp.
\]
Furthermore, for a subspace $\hilsub$ of $\hil$, let $\defpp^\hilsub$ consist of $A\in\defp$ such that $A x = 0$ for all $x\in\hil$ orthogonal to $\hilsub$, and which are strictly positive definite on $\hilsub$: if $y\in\hilsub$ is such that $\inner{y}{A y} = 0$, then $y = 0$.\footnote{Throughout the paper, $\inner{\cdot}{\cdot}$ denotes standard Euclidean inner product.} Under this notation, $B\in\defppi$ can be regarded as a $K \times K$ matrix where only the elements $B_{j \ell}$ with $(j, \ell)\in K_i \times K_i$ may be nonzero. Also note that for all $B\in\defppi$, it holds that $\proji B = B = B \proji$. For $A\in\defpp^\hilsub$ and $B\in\defpp^\hilsub$, we write $A \prec_{\hilsub} B$ to mean that $(B - A)\in\defpp^\hilsub$. Furthermore, for $B\in\defpp^\hili$, we shall denote by $B^{- \hili}$ the unique element of $\defpp^\hili$ which, on $\hili$, coincides with the unique inverse of $B$; in other words, in order to compute $B^{- \hili}$, we consider the inverse of the $K_i\times K_i$ sub-matrix and set the rest of the elements equal to zero.

For $k\in K$, define $I_k \dfn \set{i\in I \such k\in K_i}$ to be the set of traders that have access to trading security $k$. A minimal requirement for any  meaningful equilibrium model is that $|I_k| \geq 2$, for all $k\in K$. When we deal with price impact later on, we shall see that the slightly stronger condition $|I_k| \geq 3$, for all $k\in K$, is necessary and sufficient for existence (and uniqueness) of equilibrium.\footnote{The necessity of the latter assumption on linear Nash demand equilibria is well-known in the literature---see, for instance,~\cite{Kyl89} and~\cite{Viv11}.}

\subsection{Preferences and demand}\label{subsec:pref}

Let $S \equiv (S_k; \, k\in K)$ denote the payoff vector of all the tradeable securities, which are assumed to be linearly independent. In addition, we let $E_i$ denote the initial position (random endowment) of trader $\iii$. Note that we do not restrict $E_i$ to belong to the span of $S$. Following the related literature---see, for instance,~\cite{Kyl89},~\cite{Vay99} and~\cite{MalRos17}, we adapt the classic CARA-normal setting. More precisely, we assume that all traders are constant absolute risk aversion (CARA) expected utility maximisers and vector $(E_i, S)$ has a joint Gaussian law under the subjective probability $\prob_i$ of trader $i\in I$. Let $C_i$ be the covariance matrix under $\prob_i$ of $S$, where only the components of $K_i$ are regarded, and the other entries are equal to zero. Linear independence of securities implies that $C_i\in\defppi$, for each $\iii$. Furthermore, let $c_i\in\hili$ stand for the vector whose entry $k\in K_i$ is the covariance under $\prob_i$ between $E_i$ and $S_k$, and $f_i\in\hili$ denote the vector whose entry $k\in K_i$ is the expectation under $\prob_i$ of $S_k$. Traders' endowments and beliefs are considered private information.

The baseline utility (when there is no investing on $S$) of trader $\iii$ is set to be the random endowment's certainty equivalent, that is 
\[
u_i \dfn - \delta_i \log \expec_i \bra{\exp \pare{ - \frac{E_i}{\delta_i}}}\in\Real,
\]
where $\delta_i > 0$ is the risk tolerance of trader $\iii$. Then, a position $q\in\hili$ on the securities that trader $\iii$ is allowed to trade, leads to certainty equivalent equal to
\[
\hili \ni q \mapsto U_i (q) \equiv- \delta_i \log \expec_i \bra{\exp \pare{ - \frac{E_i +\inner{q}{S}}{\delta_i}}} = u_i +\inner{q}{f_i - (1 / \delta_i) c_i} - \frac{1}{2 \delta_i}\inner{q}{C_i q}.
\]
To ease the notation of the analysis that follows, we further define
\begin{equation}\label{eq:CARA_parameters}
g_i \dfn f_i - \frac{1}{\delta_i} c_i, \quad B_i \dfn \delta_i C^{- \hili}_i,
\end{equation}
where $B_i\in\defppi$ and $g_i\in\hili$. Therefore, trader $\iii$ has preferences numerically represented via the following linear-quadratic functional\footnote{The linear-quadratic functional~\eqref{eq:utility} is more general than the CARA-normal setting, in the sense that such functional could be the primary objective function of a trader, a  special case of which would be our present certainty equivalent under CARA-normal setting.}
\begin{equation}\label{eq:utility}
\hili \ni q \mapsto U_i (q) =  u_i +\inner{g_i}{q} - \frac{1}{2}\inner{q}{B^{-\hili}_i q}.
\end{equation}
Above, $g_i$ includes the trader's expectation $f_i$ of the payoffs, but also takes hedging needs into account: positive $c_i$ implies that selling securities tends to decrease the trader's risky exposure. Such hedging needs form a crucial part of the analysis, motivating traders to make transactions, even if they have homogeneous beliefs. The matrix $B_i^{-\hili}$ captures jointly the trader's risk tolerance level and the subjective covariance matrix of the securities. Finally, note that the certainty equivalent measures utility in \emph{monetary terms}, facilitating utility comparisons among different equilibria.
  
Emphasis should be given on the multi-level heterogeneity that is accommodated in this model. In particular, we allow:
\begin{itemize}
	\item heterogeneity in the traders' risk aversions;
	\item heterogeneity in the traders' subjective beliefs regarding the expectations and \emph{covariance} structure of the securities; and
	\item traders' initial random endowments which may not be spanned by the securities.\footnote{Under the standard setting of CARA expected utilities and Gaussian distributions, general initial random endowments have been considered also in non-competitive models of~\citep{Anth17, AnthrKardVich20}. Therein, it is emphasised that the beta of CAPM becomes a \emph{projected} beta, i.e., the beta of the projection of the trader’s random endowment onto the span of the tradeable securities.}
\end{itemize}

To the best of our knowledge, this is the first work in a non-competitive setting, with or without restricted participation, which allows all the above concurrently. As already stated in the introductory section, all the aforementioned heterogeneities are fairly reasonable, especially under price impact models where each trader's personal characteristics and beliefs affect the equilibrium transaction.

\subsection{Price impact}\label{sub:price_impact}

Under a competitive market setting, each trader $\iii$ is assumed to be a price-taker; therefore, for any given vector of security prices $p\in\hil$, the aim is the maximisation of the utility $U_i(q)-\inner{q}{p}$, over demand vectors $q\in\hili$. However, and as emphasised in the introductory section, there are  several security markets where such a price-taking assumption is problematic. Especially under a restricted participation environment as the one dealt with here, the possibility that large investors may influence the market is more intense, and the need arises to take into account the strategic behaviour of  participating traders in the market. 

We consider and analyse the concept of Bayesian Nash market equilibrium in linear bid schedules, as has appeared in~\cite{RosWer15} and~\cite{MalRos17}, among others. It is assumed that all traders are strategic, and that no noise traders or traders without price impact are involved in the transactions. We extend the one-round full participation game appearing in~\cite{RosWer15}, in that, in our setting, traders have heterogeneous preferences (both in risk aversions, as well as on expectations and  covariances), and do not necessarily have access to all the securities. On the other hand, the decentralised market of~\cite{MalRos17} has a richer structure on possible restrictions for trading, but all traders have the same subjective views on covariances of the securities. It is important to point out that, although traders may have access to only a subset of all securities, their actions will impact the equilibrium prices and allocations of all the securities.

Below, we give the line of argument for the individual trader's optimal allocation given a perceived price impact. We follow~\cite{Wer11} and~\cite{RosWer15}, assuming that traders perceive a linear price impact of the orders they submit; more precisely, a net order of $\Delta q\in\hili$ for trader $\iii$ will move prices by $\Lambda_i \Delta q$, where $\Lambda_i\in\defpp$ is the so-called \emph{price impact} (similar to Kyle's lambda~\cite{Kyl89}), and will be eventually endogenously determined in equilibrium.\footnote{The requirement that the price-impact matrix $\Lambda_i$ is symmetric is not without economic motivation. 
In fact, symmetric demand slope is consistent with the demand function of utility maximising traders---see, among others,~\cite[Chapter 3]{MasWhiGre}.} Let $\widetilde{p}\in\hil$ be a vector of \emph{pre-transaction} security prices. Under this linear price impact setting, an allocation $q\in\hili$ for trader $\iii$ will cost $\inner{q}{p} =\inner{q}{\widetilde{p} + \Lambda_i q}$, where $p = \widetilde{p} + \Lambda_i q$ will be the actual transaction security prices. This means that the post-transaction utility of trader $\iii$ will equal
\[
U_i (q) -\inner{q}{p} = u_i +\inner{q}{g_i - \widetilde{p}} - \frac{1}{2}\inner{q}{B^{-\hili}_i q} - \inner{q}{\Lambda_i q}.
\]
Given $\Lambda_i$, each trader $\iii$ wants to maximise the above utility by choosing demand vectors $q$ from the allowable subspace $\hili$. Therefore, with pre-transaction prices $\widetilde{p}$, the optimisation problem that trader $\iii$ faces is
\begin{equation}\label{eq:foc}
q_i = \underset{q\in\hili}{\argmax} \pare{\inner{q}{g_i - \widetilde{p}} - \frac{1}{2}\inner{q}{B^{-\hili}_i q} - \inner{q}{\Lambda_i q}}.
\end{equation}
Since the above maximisation problem is strictly concave on $\hili$, we may use first-order conditions for optimality, which give $g_i - \proji \widetilde{p} - B^{-\hili}_i q_i - 2 \proji \Lambda_i q_i = 0$; in other words,
\[g_i - \proji (\widetilde{p} + \Lambda_i q_i) = B^{-\hili}_i q_i + \proji \Lambda_i q_i  \quad \Longleftrightarrow \quad g_i - \proji p = (B^{-\hili}_i + \proji \Lambda_i \proji) q_i,\]
where the fact that $\proji q_i = q_i$ holds (since $q_i\in\hili$) was used. Note that the above first order conditions are consistent with~\cite[optimisation relation (5)]{MalRos17}, adjusted to our restricted participation setting. Since $\big( B^{-\hili}_i + \proji \Lambda_i \proji \big)\in\defppi$, upon defining
\begin{equation}\label{eq:slope_from_price_impact}
X_i \dfn \pare {B^{-\hili}_i + \proji \Lambda_i \proji}^{-\hili}\in\defppi, \quad \iii,
\end{equation}
and noting that $X_i \proji = X_i$, we obtain that
\begin{equation}\label{eq:price_impact_dem}
q_i =  X_i g_i - X_i p.
\end{equation}
To recapitulate: given a perceived linear price impact $\Lambda_i\in\defpp$, and with $X_i$ given by~\eqref{eq:slope_from_price_impact}, the relationship between the optimal allocation $q_i\in\hili$ of trader $\iii$ with actual transaction prices $p\in\hil$ is given by~\eqref{eq:price_impact_dem}. In view of~\eqref{eq:price_impact_dem}, the matrix $X_i\in\defppi$ of~\eqref{eq:slope_from_price_impact} has the interpretation of a \emph{negative demand slope} for trader $\iii$.

\subsection{Equilibrium with restricted participation and price impact}\label{sub:equil} 

Given the above best response individual trader's problem, we shall discuss now how price impact is formed in a standard uniform-price equilibrium. 

Given the relationship between the optimal allocation $q_i\in\hili$ of trader $\iii$ with transaction prices $p\in\hil$ and $X_i\in\defppi$ given by~\eqref{eq:price_impact_dem}  and~\eqref{eq:slope_from_price_impact} respectively, the market-clearing equilibrium prices $\widehat{p}$ satisfy
\[
0 = \sum_{\iii} q_i = \sum_{\iii} X_i g_i - \pare{\sum_{\iii} X_i} \widehat{p}.
\]
The following Lemma helps the analysis of several points.

\begin{lem}\label{lem:non-degen}
Suppose that $|I_k| \geq 2$, for all $k\in K$. For fixed $j\in I$, if $D_i\in\defppi$ for all $i\in I \setminus \set{j}$, then with $D_{-j} \dfn \sum_{i\in I \setminus \set{j}} D_i$, we have $D_{-j}\in\defpp^\hil$, i.e., $D_{-j}$ is invertible.
\end{lem}

From, Lemma~\ref{lem:non-degen}, it follows that
\begin{equation}\label{eq:price}
\widehat{p} = \pare{\sum_{\iii} X_i}^{-1}  \pare{\sum_{\iii} X_i g_i}.
\end{equation}
The equilibrium allocation $(\widehat{q}_i; \, \iii)$ will be given by substituting $\widehat{p}$ in~\eqref{eq:price} for $p$ in~\eqref{eq:price_impact_dem}.

Within equilibrium, each trader's perceived market impact should coincide with their actual ones; in this regard, see also~\cite[Lemma 1]{RosWer15}. Assume that all traders, except trader $\iii$, have price impacts $(\Lambda_j; \, j\in I \setminus \set{i})$, leading to $(X_j; \, j\in I \setminus \set{i})$ as in~\eqref{eq:slope_from_price_impact}. If trader $\iii$ wishes to move allocation from $\widehat{q}_i$ to $\big( \widehat{q}_i + \Delta q \big)\in\hil_i$, the aggregate position of all other traders has to change by $- \Delta q$, which would imply that new prices would equal $\widehat{p} + \Delta p$, where, by~\eqref{eq:price_impact_dem},
\[
- \Delta q + \sum_{j\in I \setminus \set{i}} \widehat{q}_j = \sum_{j\in I \setminus \set{i}} X_j g_j - \pare{\sum_{j\in I \setminus \set{i}} X_j} \pare{\widehat{p} + \Delta p}.
\]
Given that
\[
\sum_{j\in I \setminus \set{i}} \widehat{q}_j = \sum_{j\in I \setminus \set{i}} X_j g_j - \pare{\sum_{j\in I \setminus \set{i}} X_j} \widehat{p},
\]
we obtain that
\[
\Delta p = X^{-1}_{-i} \Delta q, \quad \text{where} \quad X_{-i} \dfn \sum_{j\in I \setminus \set{i}} X_j.
\]
It follows that $\Lambda_i = X_{-i}^{-1}$ has to hold in equilibrium, for all $\iii$. With the above understanding, and recalling~\eqref{eq:slope_from_price_impact}, we state the following definition of equilibrium.

\begin{defn}\label{def:equilibrium}
A collection of $(X^*_i; \, \iii)\in\prod_{\iii} \defppi$, $(q^*_i; \, \iii)\in\prod_{\iii}\X_i$ and $p^*\in\X$ will be called \textbf{Nash equilibrium} if
\begin{equation}\label{eq:Nash_impact}
X^*_i = \pare{ B_i^{-\hili} + \proji (X^*_{-i})^{-1} \proji }^{-\hili}, \quad \iii,
\end{equation}
where $X^*_{-i} \dfn \sum_{j\in I \setminus \set{i}} X^*_j$, for all $\iii$, while $p^*$ and $(q^*_i; \, \iii)$ satisfy the corresponding market-clearing condition~\eqref{eq:price} and trader's optimization solution~\eqref{eq:price_impact_dem}, i.e.,
\[p^*=\pare{\sum_{\iii} X^*_i}^{-1}  \pare{\sum_{\iii} X^*_i g_i}\quad\text{ and }\quad q^*_i=X^*_i(g_i - p^*),\quad\iii.\]

Given a Nash equilibrium as above, the \textbf{equilibrium price impacts} $(\Lambda^*_i; \, \iii)\in(\defpp)^I$ are given by $\Lambda^*_i = (X^*_{-i})^{-1}$, $\iii$.\footnote{As already mentioned, the equilibrium is consistent to the standard uniform-price double auction as the ones in~\cite{Kyl89} or~\cite{Vay99}, where traders' beliefs and endowments are considered private information. The optimal demand of each trader at equilibrium solves the point-wise optimization problem~\eqref{eq:foc} with $\Lambda_i$ determined within equilibrium, depending on other traders' demands. In particular, the optimal demand at each price is invariant on the distribution that the trader has over the endowments and the beliefs of the rest of the traders (which is private information). As pointed out also in~\cite[Section I.]{MalRos17}, a learning process will not change the optimal demand schedule of the traders, as long as the a-priori distribution that each trader has over their counterparties' private endowment and beliefs is stochastically independent to the rest of the model.}
\end{defn}

Note that Nash equilibrium as in Definition~\ref{def:equilibrium} is solely characterized by condition~\eqref{eq:Nash_impact}, since both equilibrium allocations and prices follow directly from the market-clearing condition and the traders' individual optimization problems. Therefore, the task of establishing existence and uniqueness of Nash equilibrium is focused on the solution of~\eqref{eq:Nash_impact}.

\begin{rem}[Competitive equilibrium]\label{rem:competitive}
Competitive equilibrium under heterogeneous beliefs and restricted participation will be used as a benchmark for comparisons. Letting all traders be price-takers formally means that $\Lambda_i=0$ for each $\iii$. Therefore, we obtain directly from~\eqref{eq:price} that competitive equilibrium prices equal
\begin{equation}\label{eq:competitive_price}
	\cop=\pare{\sum_{\iii} B_i}^{-1}  \pare{\sum_{\iii} B_i g_i}.
\end{equation}
The corresponding competitive equilibrium allocation is Pareto optimal (see e.g.~\cite{Viv11}), and given via~\eqref{eq:price_impact_dem}:
\begin{equation}\label{eq:competitive_allocation}
\coq_i=B_i(g_i-\cop),\qquad \forall \iii.
\end{equation}
\end{rem}

\begin{rem}[No-trading condition]\label{rem:no_trading}
We directly obtain from~\eqref{eq:competitive_price} and~\eqref{eq:competitive_allocation} that there is no trading in competitive equilibrium (i.e., $\coq_i=0$ for all $\iii$) and the initial position coincides with a Pareto optimal allocation) if and only if all $(g_i; \, \iii)$ are equal, in which case they equal $\cop$. Hence, a necessary and sufficient condition for mutually beneficial trading is traders' different hedging needs (and/or different beliefs on the expected payoffs). While traders' different estimations on the (co)variances of the tradeable assets do not form part of this condition, difference of beliefs on second moments may still create room for mutually beneficial trading. Indeed, recall from~\eqref{eq:CARA_parameters} that $g_i=f_i-c_i/\delta_i$, where $c_i$ is the vector of covariances between the trader's endowment and the tradeable assets, under the subjective probability measure $\prob_i$. In particular, even under the case where traders' endowments are the same random variable $E$, disagreement on the covariance $\cov_i(E,S)$ is sufficient to induce non-zero trading. This fact highlights the effect of such disagreement on equilibrium models. 

Interestingly enough, a similar conclusion holds in the case of the non-competitive equilibrium. Indeed, we readily get from~\eqref{eq:price_impact_dem} and~\eqref{eq:price} that Nash equilibrium allocations are all zero if and only if $(g_i; \, \iii)$ are equal, where we recall that $(X^*_i; \, \iii)$ are solely determined by $(B_i; \, \iii)$. Therefore, the necessary and sufficient condition for non-zero trading is not affected by the price impact, which however (as we shall see in the sequel) heavily affects both equilibrium quantities and utility gains, especially when traders' beliefs on the second moment deviate. 
\end{rem}

\subsection{Existence and global uniqueness of equilibrium}

Recall that, in order to have a meaningful equilibrium discussion, we assume at least two traders for every security: $|I_k| \geq 2$ holds for all $k \in K$. As Lemma~\ref{lem:no_equil} shows, if $|I_k| = 2$ holds for some $k \in K$, then there exists no Nash equilibrium in the sense of Definition~\ref{def:equilibrium}. Therefore, the stronger condition $|I_k| \geq 3$ for all $k \in K$ is \emph{necessary} for Nash equilibrium\footnote{The necessity of at least three traders for formation of non-competitive equilibrium is also mentioned in all related literature; e.g.,~\cite{Kyl89},~\cite{Vay99} and~\cite{Viv11}. Intuitively, and in the view of the discussion in \S\ref{sub:price_impact} and \S\ref{sub:equil}, three individuals are necessary because a trader's strategy impacts a price that is formed through clearing by at least two other traders.}. Theorem~\ref{thm:nash} below reveals that this condition is also \emph{sufficient} for existence of a Nash equilibrium, and that it is also unique. In previous literature, under full participation and symmetric traders with common beliefs~\citep{RosWer15} show existence and uniqueness of price-impact Nash equilibrium. In the setting of decentralised markets (again with traders of common beliefs), existence and \emph{local} uniqueness of Nash equilibrium is shown in~\citep{MalRos17}.

\begin{thm}\label{thm:nash}
Whenever $|I_k| \geq 3$ holds for all $k \in K$, a unique Nash equilibrium $(X^*_i; \, \iii)$ in the sense of Definition~\ref{def:equilibrium} exists. Moreover, for any initial collection $(X^{0}_i; \, \iii)\in\pdefppi$, if one defines inductively the updating sequence
\[
X^{n}_i \dfn \pare{B_i^{-\hili} + \proji \pare{X^{n-1}_{-i}}^{-1} \proji }^{-\hili}, \quad \iii, \quad \nin,
\] 
it holds that
\[
\limn X^{n}_i = X^*_i, \quad \forall \iii.
\]
\end{thm}

Theorem~\ref{thm:nash} not only guarantees the existence of a unique Nash equilibrium, but also provides an iterative algorithm to numerically calculate the equilibrium demands and price impacts. The only inputs for these calculations are the traders' participation restrictions sets $(K_i; \, \iii)$ and matrices $(B_i; \, \iii)$. This is a very important feature of the model, which highlights the connection between the difference of beliefs in the covariance matrix and the induced price impact.

Recall from~\eqref{eq:CARA_parameters} that the matrices $(B_i; \, \iii)$ are affected by the traders' risk aversions and beliefs on the securities covariance structure; equilibrium price impacts will not depend on the traders' beliefs on the securities' expectations and the hedging needs, that is, on the vectors $(g_i; \, \iii)$. However, these features are still important parts of the model, since they affect the after-transaction individual and aggregate utilities and hence the market's (in)efficiency. 

Note from~\eqref{eq:price_impact_dem} and~\eqref{eq:price} that Nash equilibrium is a no-trade equilibrium (i.e., $q_i = 0$ for all $\iii$) if, and only if, all vectors $(g_i; \, \iii)$ are equal. In fact, as stated in Remark~\ref{rem:no_trading}, the same necessary and sufficient non-trading condition holds for competitive equilibrium.

\subsection{The case of a risk-neutral trader}\label{sec:limit}
It is common in market microstructure literature to assume that some strategic traders are risk neutral, an assumption usually imposed to market makers or liquidity providers (see, e.g.,~\cite{Kyl85},~\cite{FarJoh02},~\cite{BiaGloSpa05} and the references therein). Although the definition of Nash equilibrium and Theorem~\ref{thm:nash} cannot be applied directly to risk neutral traders, we can accommodate the case of a single trader whose preferences approach risk neutrality via a limiting procedure.\footnote{The limiting procedure is more general than simply assuming trader's risk aversion goes to zero. We actually assume that the matrix-coefficient of the utility's quadratic term in~\eqref{eq:utility} converges to zero, which includes the case where a risk averse trader's estimations for the variances of the tradeable assets decrease.}. The 
proof of the well-defined limit is given in \S\ref{appsubsec:proof_of_prop_limiting_equil} of Appendix \ref{appsec:proofs}.

The existence of the limiting equilibrium can be seen as an extension of price-impact equilibria that have been used in~\cite{RosWer15} and~\cite{MalRos17}, so that a risk neutral trader is included in these market models. In fact, we may allow both difference of beliefs and restrictive participation in this limiting argument too.

\begin{prop}\label{prop:limiting_equil}
Let $I = \set{0, \ldots, m}$, where $m \geq 2$. Consider fixed $(B_i; \, \iii \setminus \{0 \})$, as well as a nondecreasing sequence $(B_0^n; \, \nin)$ with the property that
	$\limn (B_0^n)^{- \mathcal{X}_0} = 0$. If $(\bX^n; \, \nin)$ stands for the sequence of equilibria corresponding to $(B_0^n; \, \nin)$, then $(\bX^n; \, \nin)$  monotonically converges to a limit  $\bX^\infty\in\pdefppi$. Furthermore, $(X^\infty_i; \, \iii)$ solves the system
	\begin{equation*}
	X^\infty_0 = \pare{\proj_0 \pare{X^\infty_{-0}}^{-1} \proj_0}^{-\mathcal{X}_0}\quad\text{ and }\quad	X^\infty_i = \pare{B_i^{-\hili} + \proji \pare{X^\infty_{-i}}^{-1} \proji}^{-\hili}, \,\iii \setminus \set{0}.
	\end{equation*}
\end{prop}

Vanishing risk aversion is also related to welfare gains from the equilibrium. For example, for full participation,~\citep{MalRos17} show that traders with lower risk aversion have higher price impact in Nash equilibrium. Furthermore, under the non-competitive market models for thin risk-sharing markets studied in~\citep{Anth17},~\citep{AnthKar17} and~\citep{AnthrKardVich20}, traders with sufficiently low risk aversion prefer non-competitive markets. It turns out that this is also the case under the price-impact model at hand. Indeed, a trader with full access to the whole market who is also sufficiently close to risk neutrality prefers non-competitive to competitive equilibrium whenever the limiting transaction is non-zero (despite the fact that the aggregate utility welfare is reduced when equilibrium departures from Pareto optimal allocation). In other words, sufficiently low risk averse traders benefit by the price impact as they get higher utility at equilibrium trading. We state this (consistent to the literature) result below, where we recall the notation of the competitive equilibrium given in Remark~\ref{rem:competitive}. The proof is given in \S\ref{appsubsec:proof_prop:risk_neutral} of Appendix \ref{appsec:proofs}.

\begin{prop}\label{prop:risk_neutral}
Let $K_0=K$. Trader $0$ will be such that $\delta_0 \to\infty$, while traders $I \setminus \{ 0\}$ are fixed. Then, in the limit we have
\[\lim_{\delta_0\rightarrow\infty}\underbrace{\left\{U_0 (q^*_0) -\inner{q^*_0}{p^*}-u_0\right\}}_{\text{Utility gains at Nash}}\geq \lim_{\delta_0\rightarrow\infty}\underbrace{\left\{\left(U_0 (\coq_0) -\inner{\coq_0}{\cop}-u_0\right)\right\}}_{\text{Utility gains at competitive}}=0,\]
and, except for the (uninteresting) case $\lim_{\delta_0\rightarrow\infty}q^*_0 = 0$, the previous inequality is strict.
\end{prop}

Intuitively, the above result implies that risk neutral trader possesses power to impact the non-competitive equilibrium, since $\Lambda_0^{\infty}$ is not zero. This means that while in the competitive equilibrium a risk-neutral trader is indifferent to participate in the market in terms of utility gains, in the non-competitive setting the trader is willing to participate due to the positive utility effect of price impact. This is a clear evidence that traders with risk aversion close to zero are more willing to participate in thin markets than fully competitive ones.

\section{Heterogeneous Beliefs, Restricted Participation \& Market Efficiency}\label{sec:exampleSocialIneff}
The difference in beliefs on the covariance matrix of the tradeable assets has important implications on equilibrium price impacts, as well as on market efficiency and utility gains. The goal of this section is to highlight these implications and to separate them from the implications that are induced by simpler forms of traders' heterogeneity, such as differences in their risk aversion. We start with a quite interesting result about equilibrium pricing with and without restrictions, we then provide comparative statics, and finally give a simplified example that indicates the importance of different beliefs on second moment on the welfare comparison.

\subsection{An undesirable effect of common beliefs on second moments}

Given that tradeable assets are generally correlated, imposing participation restrictions on some of them should affect both prices and allocations of the whole market. For example, if an investor's endowment is negatively correlated to two assets, a long position on both of these assets will decrease the total risk. A restriction on trading one of these assets will intuitively imply that the long position in the unrestricted asset will increase. Under an equilibrium asset perspective, (ceteris paribus) this will increase the demand of the unrestricted asset and hence its equilibrium price. This effect should be even more pronounced in a model with price impact, in the sense that traders are not price takers, and the effect of restrictions to the their strategies (demands) should also affect prices. However, when all traders agree on the assets' second moments, this reasonable conjecture does not hold. In particular, under common beliefs on covariance matrices, restrictions on some asset leave the equilibrium prices of the unrestricted assets unchanged (although the equilibrium allocations change). In other words, disagreement on the covariance matrices is a necessary condition in order to have the reasonable effect of prices changing due to restrictions. In fact, we show that this result holds both in competitive and the non-competitive market settings. 

We start with the competitive equilibrium of Remark~\ref{rem:competitive}, and consider a market with full participation. When $C_i=C$ for all $i\in I$,~\eqref{eq:competitive_price} directly gives equilibrium prices $p^{\text{co,f}} = \sum_{i\in I} w^{\text{co}}_i g_i$, where $w^{\text{co}}_i = \delta_i / \delta$, $i\in I$, $\delta \dfn \sum_{i\in I} \delta_i$. The equilibrium price vector does not depend on the common covariance matrix; it is a weighted average of the linear parts of traders' demand functions $(g_i; \, i\in I)$. When traders agree on risk estimations, the equilibrium price reflects the source of the mutually beneficial trading, i.e., their different hedging needs, wherein the weights on the price formulation solely depend on traders' different risk appetite. On the other hand, when traders disagree on the covariance matrix of the tradeable assets, the equilibrium price vector does depend on traders' different beliefs: it is again a weighted average of traders' hedging needs, but now the weights also reflect their different estimation of risks---see again~\eqref{eq:competitive_price}. Roughly speaking, relatively lower $C_i/\delta_i$ means relatively higher weight. This is intuitive, since, as in the common-beliefs case, weights reflect traders' different risk appetite, which now includes the different estimations. 

As we shall argue in Remark~\ref{rem:same_prices_competitive} below, restricted participation does not alter the price independence on common variance: the effect of restricted traders to the prices of the unrestricted assets remains the same, even if their demands for some other assets is forced to be zero. More precisely, the equilibrium prices of the unrestricted assets remain the weighted average of (all) traders' hedging needs and the weights are again traders' relative risk tolerances. This is mainly due to the agreement on covariance matrix. If $C_i$'s do not coincide, the weights of the prices depend on the ``restricted'' variances $\pi_iC_i\pi_i$ (not the whole $C_i$), which means that compared with the full participation setting restrictions generally change the weights and hence the equilibrium prices. 

Surprisingly, including price impact in the market model does not change this fact. Similar to the competitive case, and in view of~\eqref{eq:price}, Nash equilibrium prices can be regarded as a weighted average of traders' hedging needs, with weights weights depend on equilibrium demand slopes $(X^*_i; \, i\in I)$, that are strongly linked to each trader's price impact. If traders agree on the covariance matrix, these weights again depend solely on traders' risk tolerances and not on this common covariance matrix (however, these weights are not the relative risk tolerances, as shown in~\cite{RosWer15} and~\cite{MalRos17}). As we shall see in Theorem \ref{thm:equal_prices} and the discussion following it, and again similar to the competitive case, restrictions on some assets keep weights unchanged for the unrestricted assets. This is a highly unexpected, and in our view undesirable result: when traders are not price-takers, their strategies that impact prices should depend on the pool of the assets they are allowed to trade. As it turns out, the unrestricted assets' prices will depend on the restrictions only if we allow traders to have different beliefs on the asset's second moments.

Below we formally prove the points made in the above discussion for both competitive and non-competitive markets. We assume that all traders agree on the covariance matrix: $C_i = C$ for $i\in I$, but we allow potentially different risk tolerance levels $(\delta_i; \, i\in I)$. For the restricted participation set-up, define
\[
	K_c	\dfn \bigcap_{i\in I} K_i; \quad \hil_c \dfn \bigcap_{i\in I} \hili.
\]
In words, $K_c$ contains the common assets that can be traded by everyone (i.e., the unrestricted assets), and $\hil_c$ is the subspace of $\hil \equiv \Real^K$ containing points $x\in\hil$ such that $x_j = 0$ for all $j \notin K_c$. We shall show that equilibrium prices of the unrestricted assets (those that belong in $K_c$) remain the same when we pass from full to restricted participation. As the next remark implies, this holds both in competitive and Nash equilibria, even though the prices will be different in these two different scenarios.

To ease notation throughout, we set $\pi_c$ the projection to $\hil_c$, and $\zeta = \id_\hil - \pi_c$ to be the projection on the orthogonal complement $\hilsub$ of $\hil_c$, i.e.,
\[
	\hilsub	\dfn \set{x\in\hil \such x_j = 0, \ \forall j\in K_c }.
\]

\begin{rem}[Competitive equilibrium]\label{rem:same_prices_competitive}
	Recall Remark~\ref{rem:competitive}. Under competitive equilibrium setting, $X^f_i = \delta_i C^{-1}$ holds for all $i\in I$ in the full participation case, which implies (see also~\eqref{eq:competitive_price}) that $p^f:=p^{\text{co,f}} = \sum_{i\in I} w^{\text{co}}_i g_i$ (recall that $w^{\text{co}}_i = \delta_i / \delta$, $i\in I$, $\delta \dfn \sum_{i\in I} \delta_i$). In restricted participation, straightforward but slightly tedious computations given in \S\ref{appsubsec:calc_rem:same_prices_competitive} of Appendix \ref{appsec:proofs} give prices $\pi_c p^{\text{co,r}} = \sum_{i\in I} w^{\text{co}}_i \pi_c g_i$ for assets in $K_c$, which agree with the full participation equilibrium prices for these assets.		
\end{rem}

We now consider equilibrium with price impact. Here, in the full participation case we have $X^f_i = \eta_i C^{-1}$, where the coefficients $\eta_i$ are the unique solutions to the system\footnote{The fact that this system has a unique solution comes from~\cite{MalRos17}. It also follows from Theorem~\ref{thm:nash}, in the trivial ``single asset'' case.}
\[
	1/ \eta_i = 1/\delta_i + 1/\eta_{-i}; \quad i\in I.	
\]
Indeed, it is immediate to see that the equations
\[
(X^f_i)^{-1} = \delta_i^{-1} C + (X^f_{-i})^{-1}; \quad i\in I,
\]
are satisfied at the equilibrium. We therefore have $p^f = \sum_{i\in I} w_i g_i$, where $w_i = \eta_i / \eta$, $i\in I$, $\eta \dfn \sum_{i\in I} \eta_i$.

A result similar to the one of Remark~\ref{rem:same_prices_competitive} under restricted participation comes as a corollary of Theorem~\ref{thm:equal_prices} below, the proof of which is given at \S\ref{appsubsec:proof for prices} of Appendix \ref{appsec:proofs}. Before stating the result, we introduce some additional notation: With $\hil_c$, $\pi_c$, $\hilsub$ and $\zeta$ as above, we define
\[
	\hilsub_i \dfn \hilsub \cap \hil_i; \quad i\in I,	
\]
and set $\zeta_i$ to be the projection on $\hilsub_i$, so that $\pi_i = \pi_c + \zeta_i$ for all $i\in I$. Also, we define
\[
	D \dfn 	\zeta (C - C C^{-\hil_c} C) \zeta\in\defppsub; \quad E \dfn C^{-\hil_c} C \zeta.
\]
Note that from Theorem~\ref{thm:nash}, we know that there exists a unique solution $(Y_i; \, i\in I)\in\pdefppsubi$ of the system
\[
	Y_i = (\zeta_i (\delta_i^{-1} D + Y_{-i}^{-\hilsub}) \zeta_i)^{-\hilsubi}; \quad i\in I.\footnote{In the case where $\hilsub_i = \{ 0 \}$ (i.e., when $K_i = K_c$) for $i\in I$, we tacitly assume that $Y_i \equiv 0$; then, the equation for trader $i\in I$ may be removed from the system.}	
\]

\begin{thm}\label{thm:equal_prices}
 Under the above assumptions and notation, the Nash equilibrium at the restricted participation case is determined by $(X^r_i; \, i\in I)\in\pdefppi$ which satisfies
\[
	X^r_i = \eta_i C^{-\hil_c} + E Y_i E' - E Y_i - Y_i E' + Y_i; \quad i\in I.
\]
\end{thm}

As a corollary of Theorem \ref{thm:equal_prices}, straightforward computations in \S\ref{appsubsec:calc_after_thm:equal_prices} of Appendix \ref{appsec:proofs} give prices $\pi_c p^r = \sum_{i\in I} w_i \pi_c g_i$ for assets in $K_c$, agreeing with the full participation case.

\subsection{Comparative statics}

The main input of our market model is the traders' covariance matrices, properly scaled with their risk tolerance coefficients. In this context, the following result states that equilibrium price impacts are monotonically increasing (in positive-semidefinite order) with respect to these covariance matrices; the proof follows directly from Lemma~\ref{lem:equilibrium_monotonicity}.\footnote{This result is in line with~\cite[Theorem 2, item (ii)]{MalRos17}, but we further consider heterogeneous covariance matrices, and not just different risk aversions.} 
\begin{prop}\label{pro:equilibrium_monotonicity}
Let $\bB^1 = (B^1_i; \, \iii)\in\pdefppi$ and $\bB^2 = (B^2_i; \, \iii)\in \pdefppi$ be such that $\bB^1 \preceq \bB^2$. If $\mathbf{\Lambda}^1 = (\Lambda^1_i; \, \iii)\in(\defpp)^I$ and $\mathbf{\Lambda}^2 = (\Lambda^2_i; \, \iii)\in(\defpp)^I$ stand for the associated unique equilibrium price impacts, then $\mathbf{\Lambda}^2 \preceq \mathbf{\Lambda}^1$. 
\end{prop}
 
Recall that $B_i = \delta_i C_i^{-\hili}$, for $\iii$. Therefore, if one trader believes in increased market variance (keeping risk tolerances fixed), equilibrium price impacts increase even for the assets the trader does not have access to and, more importantly, the same happens for all other traders as well. This holds under any market's limited participation setting and yields that there is higher price impact in the market where estimated risk is higher (all else being equal). Intuitively, higher $(C_i; \, \iii)$, i.e., lower $(B_i; \, \iii)$, implies lower demand slopes for all traders. The latter implies less elastic demand functions, hence traders require higher price compensation in order to offset risky positions, yielding higher price impact for all. Similar reasoning applies when subjective variances remain the same and the traders' risk tolerances $(\delta_i; \, \iii)$ increase. 


While the aforementioned monotonicity holds even when traders have common beliefs, differences of beliefs on payoffs' (co)variances do affect the comparison of traders' price impact within equilibrium. Based on Proposition~\ref{pro:equilibrium_monotonicity} and related literature, it is reasonable to expect that traders with less elastic demand functions have higher price impact at equilibrium, when compared with their more risk-averse counterparties. This is indeed the case when the covariance matrix is common for all traders and traders' heterogeneity stems from different risk aversions; in particular, traders with lower risk aversion have higher price impact at equilibrium---see~\cite[Theorem 2]{MalRos17}. However, such monotonicity does \emph{not} occur necessarily when the market has at least two assets and the heterogeneity of the demand function's slope also involves the covariance matrices: traders with less elastic demand functions do not necessarily have higher price impact, when they disagree on the assets' second moments. This situation is demonstrated through an indicative counterexample, developed in \S\ref{appsubsec:proof_of_prop:monotonicity} of Appendix \ref{appsec:proofs}. We formally state this result. 

\begin{prop}\label{prop:monotonicity}
	Let $K_i =K$, for all $i\in I$ and $|I_k| \geq 3$ for each $k$. Then, even if traders have the same risk aversion, $B_i\preceq B_j$ does not necessarily imply that $\Lambda^*_j \preceq \Lambda^*_i$ for $(i, j)\in I \times I$, where $(\Lambda_i^*; \, \iii)\in(\defpp)^I$ stands for the unique equilibrium price impacts. 
\end{prop}

If traders agree on the assets' second moments, less risk-averse traders face more elastic residual demand functions and hence their counterparties have lower price impact (ceteris paribus). Intuitively, more risk tolerant traders are more willing to offset risky positions (lower required cash compensation), which implies the impact that their counterparties' demand has on the prices to be lower. However, this situation changes when traders disagree on the second moments of a market with at least two assets. In particular, if (say) trader 1 has lower estimation of the risk in terms of covariance matrix than (say) trader 0, it is not necessarily true that trader 1 faces a more elastic residual demand than trader 0 at equilibrium. This is because the residual demand that traders face depends on the sum of other traders' estimated variance-covariance matrices and the price impact is not proportional to a common matrix, as in the case of common covariance matrix. The counterexample that we provide in order to prove Proposition~\ref{prop:monotonicity} considers two tradeable assets and focuses on trader 0 who estimates them as uncorrelated, while trader 1 estimates lower variances for both assets and also estimates negative correlation, resulting in $C_1\preceq C_0$. Therein, we show that, when we include yet one more trader having covariance matrix with positive estimated correlation, the price impact of trader 1 is not higher than that of  trader 0. This is partially because opposite estimated correlations of traders 1 and the extra trader does not imply that the residual demand that trader 0 faces is less elastic than trader 1. (It is also not true that trader 0 has higher price impact than trader 1; ordering of symmetric matrices is not a complete relation.)

\subsection{Under different beliefs, restrictions may increase welfare}\label{subsec:example}

As mentioned in the Introduction, in the presence of price impact there exist circumstances for which there is social benefit under participation restrictions in the market, as compared to the unrestricted case. The differences of beliefs in the covariance structure of the securities' returns is a crucial part of this phenomenon. The following example, where the equilibrium quantities are explicitly calculated, provides intuition on how disagreement among traders on the covariance matrix may increase aggregate utility when restrictions are imposed. 

Although a related result appears also under other forms of market's frictions (such as market decentralization of), we show here that the heterogeneity on second moments' beliefs can create welfare-increasing restrictions even under a minimal set-up, i.e., the least number of assets and traders, common risk aversions and disagreement solely on correlation. Furthermore, the welfare loss of utility from the restriction's withdrawn does not coincide with higher price impact of the trader for which the restrictions are withdrawn. 

\smallskip

For the purposes of the whole \S\ref{subsec:example}, we assume two securities and four traders: $K = \set{1,2}$, $I = \set{0,1,2,3}$, and that $\delta_i = 1$ holds for all $i\in I$. Traders in $I_{-0} = \set{1,2,3}$ are assumed identical, and such that $g_i = 0$ and $C_i = C_1$ for $i\in I_{-0}$, where
\[
C_1 = \left(\begin{array}{c c}
	1 & 0 \\
	0 & 1
	\end{array} \right).
\]
Trader 0 agrees with the rest of the traders on the assets' variance, but has different beliefs on the securities' correlation, i.e., 
\[
C_0 = \left(\begin{array}{c c}
1 & \rho \\
\rho & 1
\end{array} \right),
\]
with $\rho\in(-1,1)$. Lastly, we set $g_0 = (\gamma_1, \gamma_2)\in\Real^2\setminus \{0\}$ and recall that vector $g_0$ takes into account not only the subjective expectations of security payoffs, but also their covariance with the trader's initial endowment. For instance, positive values for $\gamma_1$ do not necessarily mean that trader 0 is more optimistic about the first security than the rest of the traders, and may simply reflect traders' difference in hedging needs. Note also that since traders in $I_{-0}$ have equal $g_i$'s, they will not trade between each other, rather they will get utility gain from the transaction through equally offsetting the hedging needs of trader 0. We will show below that in such simply thin market, the price impact of the trader 0 (who has the highest need for trading) is lower if unrestricted to trade both assets. This creates cases where the total gain of utility is lower in the setting of full participation.

\subsubsection{Restricted participation}

We assume that traders in $I_{-0}$ have no trading restrictions ($K_i = K$, for $i\in I_{-0}$), while trader 0 is restricted to trade only the first asset ($K_0 = \set{1}$). Write $(X^r_i; i\in I)$ for the solution to the system~\eqref{eq:Nash_impact} of equations, where the superscript ``$r$'' denotes \emph{restricted} participation. By symmetry, $X^r_1 = X^r_i$ holds for all $i\in I_{-0}$. Noting that $X^r_1(i, j) = 0$ whenever $(i, j) \neq (1,1)$, 
we obtain the equations
\begin{align*}
	1 / X^r_0(1,1) &= 1 + (1/3) (1 / X^r_1(1,1)), \\
	X^r_1 &= \left( C_1 + (X^r_0 + 2 X^r_1)^{-1} \right)^{-1}.
\end{align*}
In fact, one may solve the above explicitly, directly checking that the (unique) solution of these equations and the induced equilibrium price impacts are
\[
X^r_0 = \left(\begin{array}{c c}
2/3 &0 \\
0 &0
\end{array} \right), \quad
X^r_1 = \left(\begin{array}{c c}
2/3 &0 \\
0 &1/2
\end{array} \right),\quad
\Lambda^r_0 = \left(\begin{array}{c c}
1/2 &0 \\
0 & 2/3
\end{array} \right), \quad
\Lambda^r_1 = \left(\begin{array}{c c}
1/2 &0 \\
0 &1
\end{array} \right).
\] 

From~\eqref{eq:price}, the equilibrium price vector under the restricted participation setting is $p^r = (X^r_0 + 3 X^r_1)^{-1} X^r_0 g_0 = (\gamma_1/4, 0)$. Also~\eqref{eq:price_impact_dem} yields the equilibrium position $q^r_i = - X^r_i p^r = - (\gamma_1/6, 0)$ for  $i\in I_{-0}$, and a position $q^r_0 = - 3 q^r_1 = (\gamma_1/2, 0)$ for trader $0$; when $\gamma_1>0$, trader 0 buys the first security (equally) from the rest of the traders at a positive price (required cash compensation). Note that there is no transaction for the second asset, since all the unrestricted traders are identical, leaving no room for sharing risks or beliefs (implying that disagreement on assets' correlation does not affect equilibrium in the restricted market). 

Following related literature, we measure the market's efficiency with \emph{aggregate} utility. It readily follows that, at equilibrium in the restricted participation setting, this quantity equals
\[
\sum_{i\in I} \pare{\inner{q^r_i}{g_i} - \frac{1}{2 \delta^r_i}\inner{q^r_i}{C_i q^r_i}} =\inner{q^r_0}{g_0} - \frac{1}{2}\inner{q^r_0}{C_0 q^r_0} -3 \frac{1}{2}\inner{q^r_1}{ q^r_1} = \frac{\gamma_1^2}{3}.
\]
As expected, higher equilibrium transaction leads to higher trading welfare. The individual after-transaction utility gain is $\gamma_1^2/4$ for trader 0 and $\gamma_1^2/36$ for each trader $i\in I_{-0}$. Despite restrictions, trader 0's motive to trade leads to higher utility gain compared to the rest of the traders, who just offset trader 0's demand.

\subsubsection{Full participation}\label{subsubsec:example_full_part}

Here, we lift the restriction that trader 0 faced, assuming $K_i = K$ for all traders $i\in I$. Write $(X^f_i; i\in I)$ for the solution to the system of equations, where the superscript ``$f$'' denotes \emph{full} participation. This full participation equilibrium is a more complicated problem, but we shall derive an analytic expression below. By symmetry, $X^f_1 = X^f_i$ for all $i\in I_{-0}$, and we have the equations \begin{equation}\label{eq:full_part_proof1}
	X^f_0 = \left(C_0 + (3 X^f_1)^{-1} \right)^{-1}; \qquad X^f_1 = \left(C_1 + (X^f_0 + 2 X^f_1)^{-1} \right)^{-1}.
\end{equation}
Substituting the first to the second, we obtain
\begin{equation}\label{eq:full_part_proof2}
(X_1^f)^{-1} =  C_1 + ((C_0 + (3 X^f_1)^{-1})^{-1} + 2 X^f_1)^{-1}.
\end{equation}
Hence, from the solution of~\eqref{eq:full_part_proof2}, we also have the value of $X^f_0$ from~\eqref{eq:full_part_proof1}. One can find a closed-form expressions for $X_1^f$ in \eqref{eq:full_part_proof2}, and therefore obtain all relevant quantities (in particular, prices $p^f$, price impacts $\Lambda^f_i$, optimal positions $q_i^f$ for $i \in I$, as well as aggregate utility) in this full participation equilibrium. All these quantities, which are explicitly expressed in terms of $\rho$ and $g_0$, are not very indicative, and their calculation is deferred to \S\ref{appsubsec:equilibrium_example_full_part} of Appendix~\ref{appsec:proofs}. However, their usefulness emerges when we wish to compare the efficiency of the equilibria under full and restricted participation, which we tackle next.


\subsubsection{Efficiency comparison}
Equipped with analytic formulas for the equilibrium price and quantities in both full and restricted participation, we are able to compare the utility differentials and identify situations where restricted participation yields higher welfare. In such scenarios, it is socially optimal to impose trading restrictions, depending on the difference of beliefs regarding correlation and traders' hedging needs. The proof of the following result is given in \S\ref{appsubsec:proof_lem_full_part_1} of Appendix \ref{appsec:proofs}.

\begin{lem}\label{lem:full_part_1}
	Within the present setting, and with $r \dfn (2/3) \sqrt{\sqrt{113} -9} \approx 0.851$, whenever $|\rho| < r$ and $\rho \neq 0$, there exists $g_0$ which results in higher aggregate utility within the restricted, as compared to full, participation equilibrium.
\end{lem}


Through this above example, we present a situation where lifting participation restrictions for one trader results in lower aggregate utility gains. One may conjecture that this utility loss stems from the fact that this specific trader has higher price impact under full participation, the exploitation of which reduces the transaction's welfare. However, this is not the case as we now state; the proof of Lemma~\ref{lem:full_part_2} below can be found in \S\ref{appsubsec:proof_lem_full_part_2} of Appendix \ref{appsec:proofs}. 
\begin{lem}\label{lem:full_part_2}
	Within the present setting, for all $\rho\in(-1, 1)$ it holds that
$\Lambda_0^f \preceq \Lambda_0^r$.
\end{lem} 

Intuitively, Lemma \ref{lem:full_part_2} can be explained in the following way. As already mentioned, trader 0 has motive to trade with the rest of the traders, who collectively take the opposite position. When trader 0 is restricted, the difference of beliefs on the assets' correlation does not affect equilibrium prices, allocations and price impact, which is symmetrically allocated between trader 0 and all traders in $I_{-0}$ (this is also because traders in $I_{-0}$ do not trade the second asset). However, when trader 0 is allowed to trade the second asset, a different estimated correlation stands against the beliefs of traders in $I_{-0}$; this difference gives them more room to apply pressure against trader 0, since they are more in number and their price impact works aggregately against trader 0. Note that trader 0 has lower price impact for every $\rho\in(-1,1)$: it is the existence of the disagreement which matters, and not its direction or size.

Lower price impact of trader 0 at full participation causes reduction of welfare when restrictions are withdrawn, at least for some initial risk exposures. The counterexample indicates first that \emph{it is possible that full participation reduces the price impact of a trader who has more need to trade (benefits more from risk-sharing)}, and secondly that \emph{the total welfare may be lower due to the withdrawal of restrictions}. Importantly, a trader's lower price impact with less restriction (that potentially leads to restricted markets with higher welfare) may only occur when traders disagree on second moments.

\begin{rem}\label{rem:competitive_comparison}
In principle, two features affect the market's welfare differential due to restricted participation: non-competitiveness (stemming from the traders' price impact) and heterogeneity of beliefs, especially on second moments. As stated in Remark~\ref{rem:competitive}, competitive equilibrium in full participation leads to Pareto optimal allocations, regardless of traders' heterogeneity; therefore, welfare can only decrease upon imposing any kind of restricted participation in a competitive setting. On the other hand, the previous counterexample demonstrates that even the slightest heterogeneity in beliefs on second moments alone may render restricted participation beneficial in welfare terms, under non-competitive market structures.

The effect of traders' heterogeneity within a full-participation structure, either competitive or non-competitive, is not straightforward to analyse. For example, there seem to be no general monotonicity implications between the aggregate utility gains from trading and the deviation of traders' beliefs on second moments. The implications of such heterogeneity on the welfare of a competitive market needs further analysis that is beyond the scope of this work. 
\end{rem}

\section{Conclusive Remarks}\label{sec:conclu}

In this paper we consider a thin financial market, where traders are potentially restricted to trade some of the tradeable assets. Traders are not price-takers; each one strategically exploits their price impact, rendering the market non-competitive. Our main contribution to the related literature is the imposed traders' heterogeneous beliefs on the covariance matrix of the tradeable assets. As mentioned in the introduction, empirical evidence has indicated that traders indeed disagree on the (co)variances of tradeable assets. 

We first prove existence and global uniqueness of a Nash equilibrium for a general restriction structure. We also provide an algorithm that numerically calculates the equilibrium quantities. 

We then focus on traders' heterogeneity on covariance matrices, and argue that this is far form being just a theoretical venture. Different beliefs on second moments result in several economic deviations from the corresponding equilibrium under common beliefs, but where other sources of heterogeneity are present. First of all, we show that assuming common beliefs on second moments implies that restrictions on the trading of certain assets keep do not affect equilibrium prices on unrestricted assets. This is clearly an unrealistic modelling outcome, especially when assets are correlated. When traders disagree on second moments, this unfortunate feature disappears. Therefore, a prediction of our model is that restrictions change equilibrium prices only when traders disagree on the variance-covariance matrices.

Furthermore, we argue that traders with lower estimated variances do not necessarily have higher price impact in the non-competitive market, a result that comes in sharp contrast to models where heterogeneity appears in the traders' risk aversion, to which the equilibrium price impact is monotonically linked. 

Through a simple counterexample we show that, under a non-competitive market structure, restrictions may increase the market's welfare if equally risk averse traders have different views regarding the assets' covariance matrix. While this example does not provide general conditions rendering restrictions socially optimal, it gives an indicative situation when this happens. First, traders' disagreement creates conditions under which withdrawn restrictions for a certain trader reduces their price impact. Second, it is possible that when this trader is the one benefiting most from trading, aggregate utility gets reduced in full participation. The example also highlights that non-competitive market's welfare is generally very sensitive with respect to traders' disagreement on second moments, a fact that stems from the minimum possible dimension and level of disagreement imposed in the example.

The above results clarify that heterogeneity on second moment is structurally different than the plainer heterogeneity on traders' risk aversion. Even though both these quantities appear as factors affecting the position's risk under quadratic preferences, the effect of risk aversion to the Nash equilibrium is quantitatively and qualitatively different from the corresponding effect of the covariance matrix.

\smallskip

The results presented in the previous sections could have implications on thin markets (where relatively few traders who have price impact), and when considering imposing or withdrawing trading restrictions. When the market's efficiency and increase of its social welfare are concerned, our findings are linked to the policy maker decisions. 

In competitive markets, the market clearing equilibrium will result in a Pareto optimal allocation; this implies that any imposed restrictions on traders' pool of available assets will decrease the market's welfare. This conclusion is no longer valid when the market is thin and all traders act strategically against their price impact to the market. Therein, and for the first time in the literature, we point out that a very important factor for a thin market's welfare is the traders' disagreement on the second moment of the tradeable assets. Under this disagreement, withdrawal of restrictions doesn't necessarily increase the welfare. In particular, while difference on risk aversion also affects the market's welfare, it turns out that difference of beliefs on the second moment is a more influential feature when studying market's welfare. Note also that, although risk aversion is supposed to be a personal characteristic of each trader and independent on the tradeable assets, estimations on covariance matrices are---by definition---market-specific. This implies that, even when traders have the same risk aversion, a negative relation between restrictions and market welfare is not guaranteed for each market when traders disagree on second moments.

While there is no clear relation between traders' heterogeneity and welfare effect of market restrictions, the message to a policy maker is apparent. Assuming that a market is thin, then even if traders have similar risk aversion, disagreements about even a single component of the covariance matrix could make a withdrawal of market's restriction welfare-deteriorating. In fact, such situation could happen even in the minimum possible market setting, a result that indicates the high sensitivity of the relation between the restrictions and the welfare. 

The above imply that a more effective road to higher market's welfare is the improvement of market's competitiveness (such as the decrease of all traders' price impact), rather than simple separated withdrawals of participation restrictions.

\appendix
\section{Proofs}\label{appsec:proofs}

\subsection{Proof of Lemma~\ref{lem:non-degen}}
Set $D_{-j} \dfn \sum_{i\in I \setminus \set{j}} D_i$. Let $z\in\hil$, and assume $\inner{z}{D_{-j} z} = 0$. Then, $\inner{z}{D_i z} = 0$ for all $i\in I \setminus \set{j}$. Since $D_i\in \defppi$, we have $z_\ell = 0$ for all $\ell\in K_i$, whenever $i\in I \setminus \set{j}$. Therefore, $z_\ell = 0$ for all $\ell\in\bigcup_{i\in I \setminus \set{j}} K_i$. But, $\bigcup_{i\in I \setminus \set{j}} K_i = K$, since we assume that $|I_k| \geq 2$ for all $k\in K$.

\subsection{The fixed point equation}\label{appsubsec:fixed_point_equation}

The proof of Theorem~\ref{thm:nash} will be given in a series of subsections, starting with the present \S\ref{appsubsec:fixed_point_equation} and concluding in \S\ref{appsubsec:convergence_iter}.

Let $F : \pdefppi \mapsto \pdefppi$ be defined via
\begin{equation}\label{eq:functional}
F_i (\bX) = \pare{ B_i^{-\hili} + \proji (X_{-i})^{-1} \proji }^{-\hili}, \quad \iii.
\end{equation}
According to Definition~\ref{def:equilibrium}, the equilibrium negative demand slopes are given as the fixed points of $F$, i.e., solutions to the equation
\begin{equation}\label{eq:fixed_point}
\bX = F(\bX) \quad \Longleftrightarrow \quad X_i = \pare{ B_i^{-\hili} + \proji (X_{-i})^{-1} \proji }^{-\hili}, \quad \iii.
\end{equation}

The following lemma provides upper bounds for the functional $F$.
\begin{lem}\label{lem:upper_bounds_iter}
For each $\iii$, it holds that $F_i (\bX) \prec_{\hili} B_i$, as well as
\[
F_i(\bX) \prec_{\hili} \proji X_{-i} \proji \preceq X_{-i}.
\]
\end{lem}

\begin{proof}
For the first part of the lemma, we readily have that $B_i^{-\hili} \prec_{\hili} B_i^{-\hili} + \proji (X_{-i})^{-1} \proji$, which implies the order
\[
F_i(\bX) = \pare{B_i^{-\hili} + \proji (X_{-i})^{-1} \proji}^{-\hili} \prec_{\hili} B_i,\quad\forall\iii.
\]

For the second order, we first show that 
\begin{equation}\label{eq:proof1}
\pare{\proji X_{-i} \proji}^{- \hili}  \preceq  \proji (X_{-i})^{-1} \proji
\end{equation}
holds for each $\iii$. Indeed, upon rearranging the columns and rows of $X_{-i}$ bringing the sub-matrix corresponding to $K_i$ on the left top, write $X_{-i}$ and $X_{-i}^{-1}$ in block format as
\[
X_{-i} = \pare{ \begin{matrix}
	A & C \\ 
	C' & B
	\end{matrix}
				}, \qquad
X_{-i}^{-1} = \pare{ \begin{matrix}
	D & F \\ 
	F' & E
	\end{matrix} }
\]
where $A$ and $D$ are $(K_i \times K_i)$-dimensional. Since
\[
\pare{ \begin{matrix}
	A & C \\ 
	C' & B
	\end{matrix} } 
\pare{ \begin{matrix}
	D & F \\ 
	F' & E
	\end{matrix} } = 
\pare{ \begin{matrix}
	A D + C F' & A F + C E\\ 
	C' D + B F' & C' F + B E
	\end{matrix} },
\]
the fact that $X_{-i} X_{-i}^{-1}$ is the identity matrix gives
\[
A F + C E = 0 \  \Rightarrow \  F = - A^{-1} C E; \quad D  = A^{-1} - A^{-1} C F' = A^{-1} + A^{-1} C E C' A^{-1} \succeq A^{-1}.
\]
We then obtain~\eqref{eq:proof1}, since
\[
\pare{\proji X_{-i} \proji}^{-\hili} =
\pare{ \begin{matrix}
	A^{-1} & 0\\ 
	0 & 0
	\end{matrix} }  \preceq 
\pare{ \begin{matrix}
	D & 0\\ 
	0 & 0
	\end{matrix} }
= \proji (X_{-i})^{-1} \proji.
\] 
Then, it follows from~\eqref{eq:proof1} that $\pare{\proji X_{-i} \proji}^{- \hili}  \prec B^{-\hili} + \proji (X_{-i})^{-1} \proji$. The latter gives that 
\[
F_i(\bX) = \pare{B_i^{-\hili} + \proji (X_{-i})^{-1} \proji}^{-\hili} \prec_{\hili} \proji X_{-i} \proji.
\]
\end{proof}

We can already see that there is no hope for equilibrium in the case where there exists at least one asset that can be traded by at most two traders. This result is consistent with the corresponding no-equilibrium result in two-trader markets---see, for instance,~\cite{Kyl89},~\cite{Vay99} and~\cite{Viv11}. 

\begin{lem}\label{lem:no_equil}
If $|I_k| = 2$ holds for some $k \in K$, there exists no Nash equilibrium.
\end{lem}

\begin{proof}
Suppose that $\bX^*$ is Nash equilibrium, so that $\bX^* = F(\bX^*)$, and that $I_k = \set{i, j}$ holds for some $k \in K$ and $i,j\in I$ with $i\neq j$. If $e_k\in\Real^K$ stands for the zero vector with entry $1$ only in the $k$th coordinate, then (since $\proji e_k = e_k = \projj e_k$), we get from Lemma~\ref{lem:upper_bounds_iter} that
\[
X^*_i(k, k) =\inner{e_k}{X^*_i e_k} <\inner{e_k}{X^*_{-i} e_k} = X^*_{-i}(k, k) = X^*_j(k, k).
\]
A symmetric argument shows that $X^*_j(k, k) < X^*_i(k, k)$, which leads to contradiction.
\end{proof}

\subsection{Existence of fixed points}

Taking into account Lemma~\ref{lem:no_equil}, we assume hereafter that $|I_k| \geq 3$ holds for all $k\in K$. Under that assumption and based on the characterisation of the equilibrium negative demand slopes through~\eqref{eq:fixed_point}, we first show that there always exists such an equilibrium. The next step toward this goal is to show that functional $F$ defined in~\eqref{eq:functional} is  nondecreasing. For this, we need to extend the order $\preceq $ on $\pdefppi$, by defining the order
\[
\bX \equiv (X_i; \, \iii) \preceq (Y_i; \, \iii) \equiv \bY \quad \Longleftrightarrow \quad  X_i \preceq Y_i, \quad \forall \iii. 
\]
Since $\bX \preceq \bY$ implies $X_{-i} \preceq Y_{-i}$, for all $\iii$, i.e., $\proji (Y_{-i})^{-1} \proji \preceq \proji (X_{-i})^{-1} \proji$, it follows that
\[
F_i (\bX) = \pare{ B_i^{-\hili} + \proji (X_{-i})^{-1} \proji }^{-\hili} \preceq \pare{ B_i^{-\hili} + \proji (Y_{-i})^{-1} \proji }^{-\hili} = F_i (\bY),
\]
for all $\iii$. Therefore,
\begin{equation}\label{eq:monotonicity}
\bX \equiv (X_i; \, \iii) \preceq (Y_i; \, \iii) \equiv \bY \quad \Longrightarrow \quad F(\bX) \preceq F(\bY),
\end{equation}
which means that $F$ is nondecreasing. 
Furthermore, Lemma~\ref{lem:upper_bounds_iter} gives
\begin{equation}\label{eq:unifrom_domination_above}
F(\bX) \preceq \bB, \quad \forall \, \bX \equiv (X_i; \, \iii)\in\pdefppi,
\end{equation}
where $\bB\equiv (B_i; \, \iii)$. 

Define now the two sets:
\begin{equation}
\bbL \dfn \set{\bX\in\pdefppi \such \bX \preceq F (\bX)}, \quad \bbU \dfn \set{\bX\in\pdefppi \such F (\bX) \preceq \bX}
\end{equation}
and note that $\bbL \cap \bbU$ coincides with the set of fixed points of $F$. From~\eqref{eq:unifrom_domination_above}, we obtain that $t \bB\in\bbU$, for all $t\in[1,\infty)$. The next result is complementary.

\begin{lem}\label{lem:lower_bound}
There exists $\bZ\in\pdefppi$ with $\bZ \preceq \bB$ and with the property that $r Z_i \prec_{\hili} F_i (r \bZ)$ for all $i\in I$ and $r\in(0, 1]$; in particular, $r \bZ\in\bbL$, for all $r\in(0, 1]$.
\end{lem}

\begin{proof}
Pick $\alpha > 0$ small enough so that
\[
4 \alpha \proji \preceq B_i, \quad \forall \iii.
\]
Let $Z_i = \alpha \proji$, for all $\iii$, and note that $\bZ \preceq (1/4) \bB \preceq \bB$. The fact that $|I_k \setminus \set{i}| \geq 2$ holds for all $k\in K$ implies
\[
Z_{-i} = \alpha \sum_{i\in I \setminus \set{i}} \projj \succeq  2 \alpha \id_\hil.
\]
Fix $r\in(0, 1]$. Since $B_i^{- \hili} \preceq (4 \alpha)^{-1} \proji$, we have 
\[
B_i^{-\hili} + \proji (r Z_{-i})^{-1} \proji \preceq \frac{1}{2 \alpha} \pare{\frac{1}{2} + \frac{1}{r}} \proji = \frac{2 + r}{2 r \alpha} \proji.
\]
for all $\iii$. Therefore, we obtain that
\[
F_i(r \bZ) =  \pare{ B_i^{-\hili} + \proji (r Z_{-i})^{-1} \proji }^{-\hili} \succeq \frac{2}{2 + r} r \alpha \proji  \succ_{\hili} r Z_i
\]
for all $\iii$, which in particular shows that $r \bZ \preceq F (r \bZ)$, i.e., $r \bZ\in\bbL$.
\end{proof}

In the lemma below, $F^{\circ n}$ stands for the $n$-fold composition of the function $F$ with itself.
 
\begin{lem}\label{lem:sandwich_fixed_points}
Suppose that $\bX\in\bbU$ and $\bY\in\bbL$ are such that $\bY \preceq \bX$, and form sequences $(\bX^n;\, \nin)$ and $(\bY^n;\, \nin)$ via $\bX^n \dfn F^{\circ n} (\bX)$ and $\bY^n \dfn F^{\circ n} (\bY)$ for $n\in\Natural$. Then:
\begin{enumerate}
	\item The sequence $(\bX^n; \, \nin)$ is nondecreasing, the sequence $(\bY^n; \, \nin)$ is nonincreasing, and $\bY^n \preceq \bX^n$ holds for all $n\in\Natural$; in particular, the limits $\bX^* \dfn \limn \bX^n\in\pdefppi$ and $\bY^* \dfn \limn \bY^n\in\pdefppi$ exist, and satisfy $\bY \preceq \bY^* \preceq \bX^* \preceq \bX$.
	\item It holds that $\bX^* = F(\bX^*)$ and $\bY^* = F(\bY^*)$, i.e., $\bX^*$ and $\bY^*$ are equilibrium price impacts.
	\item Whenever $\bZ\in\pdefppi$ is such that $\bZ = F(\bZ)$ and $\bY \preceq  \bZ \preceq \bX$, then $\bY^* \preceq  \bZ \preceq \bX^*$
	\end{enumerate}
\end{lem}

\begin{proof}
Applying $F$ iteratively to the inequality $\bY \preceq \bX$, and using the facts that $\bX\in\bbU$ and $\bY\in\bbL$ and the monotonicity property~\eqref{eq:monotonicity}, the claims of statement (1) immediately follow. 
	
Since the monotone limits $\bX^* \dfn \limn \bX^n$ and $\bY^* \dfn \limn \bY^n$ exist and are $\pdefppi$-valued, continuity of $F$ gives
\[
\bX^* = \limn \bX^n = \limn F(\bX^{n-1}) = F( \limn \bX^{n-1}) = F(\bX^*),
\]
and similarly that $\bY^* = F(\bY^*)$, which is statement (2).
	
Finally, take $\bZ\in\bbL \cap \bbU$ with $\bY \preceq  \bZ \preceq \bX$. Using the results and notation of statements (1) and (2) with $(\bY, \bZ)$ replacing $(\bY, \bX)$, we obtain $\bY^* \preceq \bZ^* = \bZ$. Similarly, with $(\bZ, \bX)$ replacing $(\bY, \bX)$, we obtain $\bZ = \bZ^* \preceq \bX^*$.
\end{proof}

The next result shows in particular that a \emph{globally maximal} solution to~\eqref{eq:fixed_point} exists. 

\begin{lem}\label{lem:maximal_fixed_point}
Let $\bX = \bB$. Form a sequence $(\bX^n;\, \nin)$ via $\bX^n = F^{\circ n} (\bX)$, for all $\nin$. Then, the limit $\bX^* \dfn \limn \bX^n\in\pdefppi$ exists, it holds that $\bX^* = F(\bX^*)$, and for any $\bY\in\bbL$ it holds that $\bY \preceq \bX^*$.
\end{lem}

\begin{proof}
Recall that $\bX = \bB\in\bbU$ and that there exist $\bZ\in\bbL$ with $\bZ \preceq \bX$ by Lemma~\ref{lem:lower_bound}. Therefore, Lemma~\ref{lem:sandwich_fixed_points} gives that $\bX^*$ exists and $F(\bX^*) = \bX^*$. If $\bY\in\bbL$, then $\bY \preceq F(\bY) \preceq \bB = \bX$, and again by Lemma~\ref{lem:sandwich_fixed_points} we get that $\bY \preceq \bX^*$.
\end{proof}

\subsection{Uniqueness of the fixed point}

Lemma~\ref{lem:maximal_fixed_point} in fact shows that a \emph{maximal} solution $\bX^*$ of~\eqref{eq:fixed_point} exists; combined with Lemma~\ref{lem:sandwich_fixed_points}, it follow that, whenever $\bY\in\bbL$, there exists a \emph{minimal} solution of~\eqref{eq:fixed_point} that is dominated below by $\bY$ and above by $\bX^*$. By Lemma~\ref{lem:lower_bound}, there exists $\bZ\in\pdefppi$ such that $r \bZ\in\bbL$, for all $r\in(0, 1]$. We shall show below that the minimal fixed point dominated below by $r \bZ$ coincides with $\bX^*$ for all $r\in(0, 1]$; since $r\in(0, 1]$ is arbitrary, this will also imply (global) uniqueness of the fixed point.

In the sequel, along with the maximal solution $\bX^* = (X^*_i; \, \iii)$ to~\eqref{eq:fixed_point} of Lemma~\ref{lem:maximal_fixed_point}, we take $\bZ\in\pdefppi$ as in Lemma~\ref{lem:lower_bound}, fix $r\in (0, 1]$, and let $\bX \equiv (X_i; \, \iii)$ be the minimal fixed point that is bounded below by $r \bZ$. We shall show that, necessarily, $\bX = \bX^*$.

Define $\bH \dfn \bX^* - \bX$; by statement (3) of Lemma~\ref{lem:sandwich_fixed_points}, $H_i \dfn (X^*_i - X_i)\in\defp$ holds for all $\iii$. Another application of statement (3) of Lemma~\ref{lem:sandwich_fixed_points} gives $r Z_i \prec_{\hili} F_i(r \bZ) \preceq X_i$ for all $i\in I$, which implies that we can pick a sufficiently small $\epsilon > 0$ such that $r \bZ \preceq \bX - \epsilon \bH$. Consider now the mapping
\[
	[-\epsilon, 1] \ni t \mapsto \bX (t) \dfn (\bX + t \bH)\in\pdefppi,
\]
and note that $\bX(0) = \bX$ and $\bX(1) = \bX^*$.

It follows directly from definition of functional $F$, that the mapping
\[
	(-\epsilon, 1) \ni t \mapsto \Phi(t) \equiv F (\bX(t))\in\pdefppi
\]
is twice continuously differentiable. The next result shows that it is, in fact, ``\emph{concave}''.

\begin{lem}\label{lem:concavity}
With $\Phi(t) \dfn F (\bX(t))$ for $t\in(-\epsilon, 1)$, it holds that
\[
- \frac{\partial^2  \Phi (t)}{\partial t^2}\in\pdefp, \quad \forall t\in(-\epsilon, 1).
\]
It follows that the mapping $(-\epsilon, 1) \ni t \mapsto - \partial \Phi (t) / \partial t$ is nondecreasing in the order of $(\defp)^I$.
\end{lem}

\begin{proof}
For all $t\in(-\epsilon, 1)$ and $\iii$, it holds that
\[
\frac{\partial}{\partial t} X_{-i} (t)^{-1} = - X_{-i} (t)^{-1}  \pare{\frac{\partial}{\partial t} X_{-i} (t)} X_{-i} (t)^{-1} = - X_{-i} (t)^{-1}  H_{-i} X_{-i} (t)^{-1}.
\]
Since $\Phi_i (t) = \pare{(B_i)^{-\hili} + \proj_i X_{-i} (t)^{-1} \proj_i }^{-\hili}\in\defppi$ holds for $t\in(-\epsilon, 1)$ and $\iii$, it follows that
\begin{align*}
\frac{\partial \Phi_i (t)}{\partial t}  &= - \Phi_i (t) \pare{\frac{\partial}{\partial t}  \pare{\Phi_i (t)^{-\hili}} } \Phi_i (t) \\
&= - \Phi_i (t) \pare{\proj_i \frac{\partial}{\partial t} X_{-i} (t)^{-1} \proj_i  } \Phi_i (t) \\
&= \Phi_i (t) \proj_i X_{-i} (t)^{-1}  H_{-i} X_{-i} (t)^{-1} \proj_i  \Phi_i (t) \\
&= \Phi_i (t) X_{-i} (t)^{-1}  H_{-i} X_{-i} (t)^{-1} \Phi_i (t),
\end{align*}
where the last line follows because $\Phi_i (t) \proj_i = \Phi_i (t) = \proj_i \Phi_i (t)$. Call
\[
t \ni (-\epsilon, 1) \mapsto D_i (t) \dfn  X_{-i} (t)^{-1}  H_{-i} X_{-i} (t)^{-1}\in\defp, \quad \forall \, \iii.
\]
Since $X_{-i} (t) D_i (t) X_{-i} (t)=H_{-i}$ is a constant matrix as a function of $t\in(-\epsilon, 1)$, we obtain
\[
\pare{\frac{\partial}{\partial t} X_{-i} (t)} D_i (t) X_{-i} (t) + X_{-i} (t) \pare{\frac{\partial}{\partial t} D_i (t)} X_{-i} (t) + X_{-i} (t) D_i (t) \pare{\frac{\partial}{\partial t} X_{-i} (t)}  = 0,
\]
i.e.,
\[
H_{-i} D_i (t) X_{-i} (t) + X_{-i} (t) \pare{\frac{\partial}{\partial t} D_i (t)} X_{-i} (t) + X_{-i} (t) D_i (t) H_{-i} (t) = 0,
\]
which gives
\begin{align*}
\frac{\partial}{\partial t} D_i (t) &= - X_{-i} (t)^{-1} H_{-i} D_i (t) - D_i (t) H_{-i} X_{-i} (t)^{-1} \\
&= -2 X_{-i} (t)^{-1}  H_{-i} X_{-i}^{-1} (t) H_{-i} X_{-i} (t)^{-1} \\
&= -2 D_i (t) X_{-i} (t) D_i (t).
\end{align*}
Therefore, since $\partial \Phi_i (t)/\partial t  = \Phi_i (t) D_i(t) \Phi_i (t)$,
we obtain
\begin{align*}
\frac{\partial^2 \Phi_i (t)}{\partial t^2}  &= \frac{\partial \Phi_i (t)}{\partial t} D_i(t) \Phi_i (t) + \Phi_i (t) \left(\frac{\partial D_i (t)}{\partial t}\right) \Phi_i (t) +  \Phi_i (t) D_i(t)  \frac{\partial \Phi_i (t)}{\partial t} \\
&= 2 \Phi_i (t) D_i(t) \Phi_i (t) D_i(t) \Phi_i (t) - 2 \Phi_i (t) D_i(t) X_{-i}(t) D_i(t) \Phi_i (t) \\
&= - 2 \Phi_i (t) D_i(t) (X_{-i}(t) - \Phi_i (t)) D_i(t) \Phi_i (t).
\end{align*}
From Lemma~\ref{lem:upper_bounds_iter}, we have that $\Phi_i (t) \preceq X_{-i}(t)$. Also, since $\Phi_i (t) D_i(t)\in\defp$, it is clear that
\[
- \frac{\partial^2  \Phi_i (t)}{\partial t^2}\in\defp, \quad \forall \, t\in(-\epsilon, 1) \quad \text{and} \quad \forall \, \iii.
\]

For $- \epsilon \leq t_1 < t_2 \leq 1$, the above implies that
\[
\frac{\partial \Phi_i (t)}{\partial t} \Big|_{t = t_1} - \frac{\partial \Phi_i (t)}{\partial t} \Big|_{t = t_2} =\int_{t_1}^{t_2} - \frac{\partial^2  \Phi_i (t)}{\partial t^2} \ud t\in\defp, \quad \forall \iii,
\]
completing the argument.
\end{proof}

Recall that $\bX \equiv (X_i; \, \iii)$ stands for the minimal fixed point of $F$ that is bounded below by $r \bZ$, and our aim is to show that $\bX=\bX^*$.

\begin{lem}\label{lem:uniqueness_key}
When $-\epsilon \leq s < 0$, it holds that $\bX(s)\in\bbU$, i.e., $F(\bX(s)) \preceq \bX(s)$.
\end{lem}

\begin{proof}
Note that
\[
\bH = \bX^* - \bX = F (\bX(1)) - F (\bX(0)) =\int_0^1 \frac{\partial F (\bX (t))}{\partial t} \ud t.
\]
In view of the fact that $(-\epsilon, 1) \ni t \mapsto - \partial F (\bX (t)) / \partial t$ is nondecreasing as follows from Lemma~\ref{lem:concavity}, we have
\[
\bH =\int_0^1 \frac{\partial F (\bX (t))}{\partial t} \ud t \preceq \frac{1}{s}\int_{0}^s \frac{\partial F (\bX (t))}{\partial t} \ud t = \frac{1}{s} \pare{F (\bX(s)) - F (\bX(0))}, \quad s\in(0,1]. 
\]
Therefore,
\[
\frac{\partial F (\bX (s))}{\partial s} \Big|_{s = 0} = \lim_{s \downarrow 0} \frac{ F (\bX(s)) -  F (\bX(0))}{s} \succeq \bH.
\]
Using again the fact that $(-\epsilon, 1) \ni t \mapsto - \partial F (\bX (t)) / \partial t$ is nondecreasing, which implies that $ \partial F (\bX (t)) / \partial t \succeq \bH$ holds for all $t\in(-\epsilon, 0)$, we obtain
\[
\bX - F(\bX (s)) = F(\bX(0)) - F(\bX(s)) =\int_s^0 \frac{\partial F (\bX (t))}{\partial t} \ud t \succeq -s \bH, \quad \forall \, s\in[-\epsilon, 0),
\]
which shows that $F (\bX(s)) \preceq \bX + s \bH = \bX (s)$ holds when $-\epsilon \leq s < 0$.
\end{proof}

We are now ready to complete the proof of uniqueness. Recall that $\epsilon > 0$ was picked such that $r \bZ \preceq \bX(-\epsilon) $. Since additionally $r \bZ\in\bbL$ and $\bX(-\epsilon)\in\bbU$, Lemma~\ref{lem:sandwich_fixed_points} gives existence of a fixed point $\bY$ with the property $r \bZ \preceq \bY \preceq \bX(-\epsilon)$. Since $\bX$ is the smallest fixed point dominated below by $r \bZ$, this would give $\bX \preceq \bY$ which together with $\bY \preceq \bX - \epsilon \bH$ and $\bH\in\pdefp$ gives $\bH = 0$, i.e., $\bX = \bX^*$.

\subsection{Convergence to solutions through iteration}\label{appsubsec:convergence_iter}

Now that uniqueness has been established, we can show that the iterative procedure will always converge to the unique root and hence finish the proof of Theorem~\ref{thm:nash}.

\begin{lem}\label{lem:iteration-fixed-point}
For an arbitrary $\bX^0\in\pdefppi$, form a sequence $(\bX^n;\, \nin)$ by induction, asking that $\bX^n = F (\bX^{n-1})$, for all $\nin$. Then, it holds that
\[
\limn \bX^n = \bX^*.
\]
Furthermore, if $\bX^0\in\bbL$, the sequence $(\bX^n; \, \nin)$ is nondecreasing, while if $\bX^0\in\bbU$, the sequence $(\bX^n; \, \nin)$ is nonincreasing. 
\end{lem}

\begin{proof}
If $\bX^0\in\bbL$, the inequality $\bX^0 \preceq F(\bX^0) = \bX^1$ and the monotonicity of $F$ of the form~\eqref{eq:monotonicity} show by induction that $(\bX^n; \, \nin)$ is nondecreasing. Similarly, if $\bX^0\in\bbU$, the inequality $\bX^1 = F(\bX^0) \preceq \bX^0$ and the monotonicity of $F$ show that $(\bX^n; \, \nin)$ is nonincreasing.

Given an arbitrary $\bX^0\in\pdefppi$, recall that Lemma~\ref{lem:upper_bounds_iter} implies that $t \bB\in \bbU$ for all $t\in[1,\infty)$ and that Lemma~\ref{lem:lower_bound} guarantees the existence of $\bZ\in\pdefppi$, such that $\bZ\preceq \bB$ and $r\bZ\in\bbL$ for all $r\in(0,1]$. Pick $\hat{r}\in(0,1]$ sufficiently small and $\hat{t}\in[1,\infty)$ sufficiently large such that
\[
\bW^0 \dfn \hat{r} \bZ \preceq \bX^0 \preceq \hat{t} \bB =: \bY^0
\]
holds. Then, define the sequences $(\bW^n; \, \nin)$ and $(\bY^n; \, \nin)$ by iteration via $F$, that is $\bW^n:=\bF(\bW^{n-1})$ and $\bY^n:=\bF(\bY^{n-1})$, for each $\nin$. Since $\bW^0 \preceq \bX^0 \preceq \bY^0$, monotonicity of $F$ and induction gives $\bW^{n} \preceq \bX^{n} \preceq \bY^{n}$, for each $\nin$. Also, since $\bW^0\in\bbL$ and $\bY^0\in\bbU$, the sequence $(\bW^n; \, \nin)$ is nondecreasing and, in fact, bounded above by $\bY^0$, while the sequence $(\bY^n; \, \nin)$ is nonincreasing and bounded below by $\bW^0$. It follows that both sequences have limits $\bW^{\infty}$ and $\bY^\infty$, respectively, with $\bW^\infty \preceq \bY^\infty$. Continuity of $F$ gives that $\bW^\infty = F(\bW^\infty)$ and $\bY^\infty = F(\bY^\infty)$, exactly as in the proof of Lemma~\ref{lem:maximal_fixed_point}. By uniqueness of the solution to~\eqref{eq:fixed_point}, it follows that $\bW^\infty = \bX^* = \bY^\infty$, from which it further follows that $\limn \bX^n = \bX^*$.
\end{proof}

\subsection{Proof of Proposition~\ref{prop:limiting_equil}}\label{appsubsec:proof_of_prop_limiting_equil}

The main input of our market model is the traders' covariance matrices, properly scaled with their risk tolerance coefficients. In this context, the next result generalises~\cite[Proposition 1, item (iv)]{MalRos17} by considering heterogeneous covariance matrices. It implies that equilibrium price impact is monotonically increasing (in positive-semidefinite order) with respect to these covariance matrices.

\begin{lem}\label{lem:equilibrium_monotonicity}
Let $\bB^1 = (B^1_i; \, \iii)\in\pdefppi$ and $\bB^2 = (B^2_i; \, \iii)\in \pdefppi$ be such that $\bB^1 \preceq \bB^2$. If $\bX^1 = (X^1_i; \, \iii)\in\pdefppi$ and $\bX^2 = (X^2_i; \, \iii)\in \pdefppi$ stand for the associated unique equilibria, then $\bX^1 \preceq \bX^2$. 
\end{lem}

\begin{proof}
Set $F^2$ to be as in~\eqref{eq:fixed_point} with $\bB^2$ in place of $\bB$ there, and note that
\[
X^1_i = ((B^1_i)^{-\hili} + \proji (X^1_{-i})^{-1}\proji )^{-\hili} \preceq ((B^2_i)^{-\hili} +\proji  (X^1_{-i})^{-1}\proji )^{-\hili} = F^2_i (\bX^1), \quad \iii,
\]
which shows that $\bX^1 \preceq F^2 (\bX^1)$. Then, statement (1) of Lemma~\ref{lem:sandwich_fixed_points} and uniqueness imply that $\bX^1 \preceq \bX^2$.
\end{proof}

We are now in position to complete the proof of Proposition~\ref{prop:limiting_equil}. By monotonicity from Lemma~\ref{lem:equilibrium_monotonicity} and the nondecreasing assumption of $(B_0^n; \, \nin)$, we have that $(\bX^n; \, \nin)$ is also nondecreasing in $\pdefppi$. Furthermore, from Lemma~\ref{lem:upper_bounds_iter} we have that 
\[
X_i^n \preceq B_i, \quad \forall \iii\setminus \set{0}
\]
and also that 
\[
X_0^n \preceq X^n_{-0} \preceq \sum_{\iii \setminus \set{0}} B_i.
\]
It follows that $(\bX^n; \, \nin)$ has a monotone limit $\bX^\infty\in\pdefppi$. Since $\limn (B_0^n)^{- \mathcal{X}_0} = 0$, condition~\eqref{eq:Nash_impact} in the limit gives that
\[
X_0^\infty = \limn  X_0^n  = \limn \pare{(B^n_0)^{-\hil_0} + \proj_0 (X^n_{-0})^{-1} \proj_0}^{-\hil_0} = \pare{\proj_0 (X^\infty_{-0})^{-1} \proj_0}^{-\hil_0}.
\] 
The same limiting argument shows that
\[
X^\infty_i = \pare{B_i^{-\hili} + \proji \pare{X^\infty_{-i}}^{-1} \proji}^{-\hili}, \quad \iii	
\]

\subsection{Proof of Proposition~\ref{prop:risk_neutral}}\label{appsubsec:proof_prop:risk_neutral}
Let us first consider the competitive market. The demand function for each trader $\iii$ is given as the solution of the problem~\eqref{eq:foc} when $\Lambda_i=0$ for each $\iii$. Since $B_0 = \delta_0 C_0^{-1}$ from~\eqref{eq:CARA_parameters}, equations~\eqref{eq:competitive_price} and~\eqref{eq:competitive_allocation} give that
\[
	\lim_{\delta_0\rightarrow\infty}\cop = g_0; \qquad \lim_{\delta_0\rightarrow\infty}\coq_0=-\lim_{\delta_0\rightarrow\infty}\sum_{\iii\setminus\{0\}}\coq_i=\sum_{\iii\setminus\{0\}}\delta_i B_i (g_0-g_i).
\]
It then follows that, as $\delta_0 \rightarrow\infty$,
\[
	U_0 (\coq_0) -\inner{\coq_0}{\cop} - u_0 =\inner{\coq_0}{g_0-\cop}-\frac{1}{2\delta_0}\inner{\coq_0}{C_0\coq_0}\rightarrow 0.
\]

On the other hand, we get from Proposition~\ref{prop:limiting_equil} and the assumption $K_0 = K$ that
\[
X^\infty_0 = \pare{\pare{X^\infty_{-0}}^{-1}}^{-1} = X^\infty_{-0}.	
\]  
But then $X^\infty_{-i} = 2 X^\infty_{-0} - X^\infty_i$ holds for all $\iii\setminus\set{0}$, and hence we obtain that $(X^\infty_i; \, \iii \setminus \{0 \})$ solves the system
\[
\pare{X^\infty_i}^{-\hili} = B_i^{-\hili} + \proji \pare{2 X^\infty_{-0} - X^\infty_i}^{-1} \proji, \quad \iii \setminus \set{0}.	\]
Now, by~\eqref{eq:price_impact_dem},~\eqref{eq:price} and Proposition~\ref{prop:limiting_equil}, from where it particularly follows that $(X_i^*; \, \iii)$ converge as $\delta_0 \to\infty$ and $X_0^* = X_{-0}^*$ in the limit, we obtain 
\[
	\lim_{\delta_0\rightarrow\infty}p^* = \frac{1}{2}g_0 + \frac{1}{2}(X_0^{\infty})^{-1}\sum_{\iii\setminus\{0\}} X_i^{\infty}g_i; \qquad \lim_{\delta_0\rightarrow\infty}q^*_0 = \frac{1}{2}\left(X_0^{\infty} g_0 - \sum_{\iii\setminus\{0\}}X_i^{\infty}g_i\right).
\]
Hence, it follows that
\[
	U_0 (q^*_0) -\inner{q^*_0}{p^*} - u_0 \rightarrow \frac{1}{4}\inner{X_0^{\infty}\left(X_0^{\infty} g_0-\sum_{\iii\setminus\{0\}}X_i^{\infty}g_i\right)}{X_0^{\infty}g_0-\sum_{\iii\setminus\{0\}}X_i^{\infty}g_i}.
\]
Since $X_0^{\infty}$ is positive-definite matrix, the above limit is non-negative, and equal to zero if, and only if, $X_0^{\infty} g_0-\sum_{\iii\setminus\{0\}}X_i^{\infty}g_i=0$; the latter is equivalent to $\lim_{\delta_0\rightarrow\infty} q^*_0 = 0$.

\subsection{Calculations in Remark~\ref{rem:same_prices_competitive}}\label{appsubsec:calc_rem:same_prices_competitive}\
We keep all notation from Remark~\ref{rem:same_prices_competitive}. In restricted participation, $X_i^r = \delta_i C^{-\hili}$ for all $i\in I$. We readily obtain $\pi_c (C X_i^r) \pi_c = \delta_i \pi_c = \pi_c (C X_i^r)$; therefore, $\pi_c (C X_i^r) \zeta = 0$. With $X^r \dfn \sum_{i\in I} X^r_i$, we then have $\pi_c (C X^r) \pi_c = \delta \pi_c = \pi_c (C X^r) $ and, therefore, $\pi_c (C X^r) \zeta = 0$. The last also implies that $\pi_c (C X^r)^{-1} \pi_c = \delta^{-1} \pi_c = \pi_c (C X^r)^{-1}$. (This can be readily seen by computing the inverse.) Therefore,
	\[
		\pi_c (X^r)^{-1} X_i = \pi_c (C X^r)^{-1} C X^r_i = \pi_c (C X^r)^{-1} \pi_c C X^r_i = w^{\text{co}}_i \pi_c.
	\]
This implies prices $\pi_c p^{\text{co,r}} = \sum_{i\in I} w^{\text{co}}_i \pi_c g_i$ for assets in $K_c$.

\subsection{Proof of Theorem~\ref{thm:equal_prices}}\label{appsubsec:proof for prices}\
We only need to show that the given $(X^r_i; \, i\in I)\in\pdefppi$ satisfy the system
\[
	X^r_i = (\pi_i (\delta_i^{-1} C + (X^r_{-i})^{-1}) \pi_i)^{-\hili}; \quad i\in I,	
\]
which is that same as
\[
	(X^r_i)^{-\hili} =  \delta_i^{-1} \pi_i C \pi_i + \pi_i (X^r_{-i})^{-1} \pi_i; \quad i\in I.	
\]
To show this, first note that direct computation involving inverses show that
\[
	\eta_i (X^r_i)^{-\hili} = \pi_c C \pi_c + \zeta_i C \pi_c + \pi_c C \zeta_i + \zeta_i C C^{-\hil_c} C \zeta_i + \eta_i Y_i^{-\hilsubi}.
\]
Using that $\pi_i = \pi_c + \zeta_i$ and the definition of $D$, this gives us
\[
	(X^r_i)^{-\hili} = \eta_i^{-1} (\pi_i C \pi_i - \zeta_i D \zeta_i) + Y_i^{-\hilsubi}
\]
Furthermore, since
\[
	X^r_{-i} = \eta_{-i} C^{-\hil_c} + E Y_{-i} E' - E Y_{-i} - Y_{-i} E' + Y_{-i}
\]
where $Y_{-i}\in\defppsub$, so that $X^r_{-i}$ is invertible, giving, similar to above,
\[
	\eta_{-i} (X^r_{-i})^{-1} = \pi_c C \pi_c + \zeta C \pi_c + \pi_c C \zeta + \zeta C C^{-\hil_c} C \zeta + \eta_{-i} Y_{-i}^{-\hilsub}	
\]
It follows that
\[
	\eta_{-i} \pi_i (X^r_{-i})^{-1} \pi_i = \pi_c C \pi_c + \zeta_i C \pi_c + \pi_c C \zeta_i + \zeta_i C C^{-\hil_c} C \zeta_i + \eta_{-i} \zeta_i Y_{-i}^{-\hilsub} \zeta_i'
\]
again, using $\pi_i = \pi_c + \zeta_i$ and the definition of $D$, this gives us
\[
	\pi_i (X^r_{-i})^{-1} \pi_i = \eta_{-i}^{-1} (\pi_i C \pi_i -\zeta_i D \zeta_i) + \zeta_i Y_i^{-\hilsubi}\zeta_i.
\]
Since $\delta_i^{-1} + \eta_{-i}^{-1} = \eta_i^{-1}$ holds for $i\in I$, we have
\begin{align*}
	\delta_i^{-1} \pi_i C \pi_i +  \pi_i (X^r_{-i})^{-1} \pi_i &=  \eta_i^{-1} \pi_i C \pi_i - \eta_{-i}^{-1} \zeta_i D \zeta_i + \zeta_i Y_i^{-\hilsubi} \zeta_i \\
	&= (X^r_i)^{-\hili} - Y_i^{-\hilsubi} + \eta_i^{-1} \zeta_i D \zeta_i  - \eta_{-i}^{-1} \zeta_i D \zeta_i + \zeta_i Y_i^{-\hilsubi} \zeta_i \\
	&= (X^r_i)^{-\hili} + \delta_i^{-1} \zeta_i D \zeta_i - Y_i^{-\hilsubi} + \zeta_i Y_i^{-\hilsubi} \zeta_i  = (X^r_i)^{-\hili}
\end{align*}
completing the argument.

\subsection{Calculations after Theorem~\ref{thm:equal_prices}}\label{appsubsec:calc_after_thm:equal_prices}

Given the statement of Theorem~\ref{thm:equal_prices}, straightforward computations give $\pi_c (C X_i^r) \pi_c = \eta_i \pi_c = \pi_c (C X_i^r)$ and, therefore, $\pi_c (C X_i^r) \zeta = 0$. The situation now is the same as the calculations in \S\ref{appsubsec:calc_rem:same_prices_competitive}. More precisely, with $X^r \dfn \sum_{i\in I} X^r_i$, we have $\pi_c (C X^r) \pi_c = \eta \pi_c = \pi_c (C X^r) $ and, therefore, $\pi_c (C X^r) \zeta = 0$; the latter also implies that $\pi_c (C X^r)^{-1} \pi_c = \eta^{-1} \pi_c = \pi_c (C X^r)^{-1}$, by computing the inverse. Therefore,
\[
	\pi_c (X^r)^{-1} X_i = \pi_c (C X^r)^{-1} C X^r_i = \pi_c (C X^r)^{-1} \pi_c C X^r_i = w_i \pi_c,
\]
implying prices $\pi_c p^r = \sum_{i\in I} w_i \pi_c g_i$ for assets in $K_c$.

\subsection{Proof of Proposition~\ref{prop:monotonicity}}\label{appsubsec:proof_of_prop:monotonicity}

The proof is based on an indicative counterexample. For this, we first need the following lemma. 
\begin{lem}\label{lem:charact_full}
In full participation Nash equilibrium, with $X^* \dfn \sum_{j\in I} X_j^*$, it holds that
	\[
	X^*_i = B_i + \frac{1}{2} X^* - (B_i)^{1/2} \pare{\id + \frac{1}{4} \pare{(B_i)^{-1/2} X^* (B_i)^{-1/2}}^2 }^{1/2} (B_i)^{1/2} , \quad \iii.
	\]	 
\end{lem}

\begin{proof}
	Recall that
	\[
	(X^*_i)^{-1} = (B_i)^{-1} + (X^*_{-i})^{-1}, \quad \iii,
	\]
	Define the matrices
	\[
	P_i =  \frac{1}{2} (B_i)^{-1/2} X^*_i (B_i)^{-1/2}, \quad T_i = \frac{1}{2} (B_i)^{-1/2} X^* (B_i)^{-1/2}, \quad \iii,
	\]
	and note that $P_i\in\defpp$, $T_i\in\defpp$ satisfy $2 \id \prec P_i^{-1}$ (because $X^*_i \prec B_i$), $P_i \prec T_i$ (because $X^*_i \prec X^*$), and that
	\[
	P^{-1}_i = 2 \id + (T_i - P_i)^{-1}, \quad \iii. 
	\]
	Let $P_i = Q_i^\top L_i Q_i$ be the decomposition of $P_i\in\defpp$, where $Q_i$ is orthonormal and $L_i$ is diagonal, for all $\iii$. Then, note that
	\[
	T_i = P_i + (P_i^{-1} - 2 \id)^{-1} = Q_i^\top \pare{ L_i + (L^{-1}_i - 2\id)^{-1}} Q_i, \quad \iii. 
	\]
	which implies that also $T_i = Q_i^\top M_i Q_i$ holds for a diagonal $M_i$ (the decomposition of $T_i$ is with the same $Q_i$ as that of $P_i$), and, therefore, that
	\[
	L^{-1}_i = 2 \id + (M_i - L_i)^{-1}, \quad \iii. 
	\]
	In other words, simple algebra (all matrices are diagonal) gives
	\[
	L_i^2 - (\id + M_i) L_i  + (1/2) M_i = 0, \quad \iii. 
	\]
	For concreteness, set $L_i = \diag (l_i)$ and $M_i = \diag (m_i)$; then, $0 < l_i(k) < m_i(k)$ for $k\in K$, and
	\[
	l_i^2(k) - (1 + m_i(k)) l_i(k)  + (1/2) m_i(k) = 0, \quad \iii, \quad k\in K. 
	\]
	The above quadratic equation (in $l_i(k)$) has in general the two roots
	\[
	\frac{1 + m_i(k) \pm \sqrt{1 + m_i(k)^2}}{2}, \quad \iii, \quad k\in K;
	\]
	however, the one with the ``$+$'' sign gives a value strictly greater than $m_i(k)$, which the one with the ``$-$'' sign gives a value lying in the interval $(0, m_i(k))$. It follows that
	\[
	2 l_i(k) =  1 + m_i(k) - \sqrt{1 + m_i(k)^2}, \quad \iii, \quad k\in K.
	\]
	Given that both $L_i$ and $M_i$ are diagonal, one may write this in matrix notation as
	\[
	2 L_i = \id + M_i - \pare{ I + M^2_i}^{1/2}, \quad \iii.
	\]
	Multiplying the above equation with $Q_i^\top$ from the left and $Q_i$ from the right, we obtain that
	\[
	2 P_i = \id + T_i - \pare{I + K^2_i}^{1/2}, \quad \iii.
	\]
	Recalling the definitions of $P_i$ and $T_i$, and multiplying the above equation with $(B_i)^{1/2}$ from both the left and the right, we obtain
	\[
	X^*_i = B_i + \frac{1}{2} X^* - (B_i)^{1/2} \pare{ \id + \frac{1}{4} \pare{(B_i)^{-1/2} X^* (B_i)^{-1/2}}^2}^{1/2} (B_i)^{1/2}, \quad \iii.
	\]
	This concludes the proof of Lemma~\ref{lem:charact_full}.
\end{proof}

We continue with the counterexample that will establish Proposition~\ref{prop:monotonicity}. Consider a full participation setting with three traders $I = \{0,1,2\}$ and two assets. As a baseline ``0'' model, and using superscripts to denote quantities under the model we are considering, assume that $B^0_0 = \diag(1, 2/3) = B^0_1$, while $B^0_2 = \diag(1, 12/5)$. Then, in equilibrium we have $X_0^0 = \diag(1/2, 2/5) = X_1^0$, while $X^0_2 = \diag(1/2, 3/5)$.

We now consider ``$\epsilon$'' model perturbations. In all models, trader 0 remains unchanged ($B_0^\epsilon = B_0^0$ for $\epsilon > 0$), and we shall construct equilibrium such that $X^\epsilon \equiv X^\epsilon_0 + X^\epsilon_1 + X^\epsilon_2$ is always equal to $X^0 = \diag(3/2, 7/5)$ for all sufficiently small $\epsilon > 0$. (We shall see how to accommodate this last part.) In this case, Lemma~\ref{lem:charact_full} implies that $X_0^\epsilon = X_0^0$ for $\epsilon > 0$. For trader 1, we consider
\[
	B_1^\epsilon = B_1^0 + \epsilon B_1^0 \pare{ \begin{matrix}
	1 & 1 - \epsilon \\ 
	1 - \epsilon & 1
	\end{matrix} } B_1^0 = B_1^0 + \epsilon B_1^0 C B_1^0 + o(\epsilon); \quad C := \pare{ \begin{matrix}
	1 & 1 \\ 
	1 & 1
	\end{matrix}},
\]
where in this particular formula (but not the ones below) the error term $o(\epsilon)$ equals exactly
\[
	\epsilon^2
	\pare{ \begin{matrix}
	0 & -2/3 \\ 
	-2/3 & 0
	\end{matrix}}.
\]
First of all, note that $B^\epsilon_0 \prec B^\epsilon_1$ holds for $0<\epsilon < 2$, since then the matrix
\[
	\pare{ \begin{matrix}
		1 & 1 - \epsilon \\ 
		1 - \epsilon & 1
		\end{matrix} }
\]
is strictly positive definite. Secondly, and since we aim at keeping $X^\epsilon$ constant in $\epsilon > 0$, Lemma~\ref{lem:charact_full} implies that we will have $X^\epsilon_1 = X^0_1 + \epsilon \Delta X^0_1 + o (\epsilon)$, for a matrix $\Delta X^0_1$ to be determined. Again, using the fact that we want to keep $X^\epsilon = X^0$ for $\epsilon > 0$, differentiating both sides the equality $(X_1^\epsilon)^{-1} = (B_1^\epsilon)^{-1} + (X^0 - X_1^\epsilon)^{-1}$ with respect to $\epsilon$ and taking the limit as $\epsilon\rightarrow 0$, we obtain
\[
	(X_1^0)^{-1} \Delta X^0_1 (X_1^0)^{-1} = (B_1^0)^{-1} (B_1^0 C B_1^0) (B_1^0)^{-1} - (X^0 - X_1^0)^{-1} \Delta X^0_1 (X^0 - X_1^0)^{-1}.
\]
Noting that $X^0 - X_1^0 = \id$, we have $\Delta X^0_1 + \diag(2, 5/2) \Delta X^0_1 \diag(2, 5/2) = C$. Solving for $\Delta X^0_1$, we obtain
\[
	\Delta X^0_1 = \pare{ \begin{matrix}
	1/5 & 1/6 \\ 
	1/6 & 4/29
	\end{matrix} }.
\]
For small $\epsilon > 0$, $X^\epsilon_0 \prec X^\epsilon_1$, which is equivalent to $\Lambda^\epsilon_0 \prec \Lambda^\epsilon_1$, is equivalent to $\Delta X^0_1 \succ 0$; however, the determinant of $\Delta X^0_1$ equals $4 / 145 - 1 / 36 = - 1 / 5220 < 0$. It follows that $\Lambda^\epsilon_0 \prec \Lambda^\epsilon_1$ fails in this case for small $\epsilon > 0$, which is exactly the context of Proposition~\ref{prop:monotonicity}. 

For completeness, it has to be mentioned how to keep $X^\epsilon$ constant here, so that the previous calculations are valid. With $X_0^\epsilon$ and $X_1^\epsilon$ defined as previously, we set
\[
	B^\epsilon_2 = \pare{(X^0 - X^\epsilon_0 - X^\epsilon_1)^{-1} - (X^\epsilon_0 + X^\epsilon_1)^{-1}}^{-1}; 
\]
for small $\epsilon > 0$, this will be a positive definite matrix. Then, by Theorem~\ref{thm:nash}, the unique Nash equilibrium will be such that $X^\epsilon = X^0$ for all small enough $\epsilon > 0$.

\subsection{Equilibrium in the setting of \S\ref{subsubsec:example_full_part}}\label{appsubsec:equilibrium_example_full_part}

We keep all notation from \S\ref{subsec:example}, and in particular \S\ref{subsubsec:example_full_part}. We first diagonalise $C_0$ as 
\[
C_0 = V \left(\begin{array}{c c}
	1 + \rho & 0 \\
	0 & 1 - \rho
	\end{array} \right) V, \quad \text{where} \quad
V = \frac{1}{\sqrt{2}} \left(\begin{array}{c c}
		1 & 1 \\
		1 & -1
		\end{array} \right).
\]
Note that $V$ is symmetric and unitary: $V^2 = \id$, where $\id$ stands for the $2 \times 2$ identity matrix. Therefore, $V$ also trivially diagonalises the identity matrix $C_1$. It then follows that $V X_1^f V$ will be also diagonal. To ease notation below, define the matrix
\[
D^\rho	\dfn \left(\begin{array}{c c}
		1/ (1 + \rho) & 0 \\
		0 & 1/(1 - \rho)
		\end{array} \right),
\]
so that $C_0^{-1} = V D^\rho V$, and for functions $h : (0 ,\infty) \to \Real$ write $h(D^\rho)$ for the $2 \times 2$ diagonal matrix with diagonal entries $(h(1/ (1 + \rho)), h(1/ (1 - \rho)))$. Then, one solves the matrix equation~\eqref{eq:full_part_proof2} for $X_1^f$ to obtain 
\[
X^f_1 = V h_1 (D^\rho) V, \quad \text{where} \quad h_1(x) = \frac{1}{4} - \frac{5}{12} x  + \sqrt{ \left( \frac{1}{4} - \frac{5}{12} x \right)^2 + \frac{2}{3} x }, \quad x\in(0,\infty).
\]
It then also follows from~\eqref{eq:full_part_proof1} that
\[
X^f_0 = V h_0 (D^\rho) V, \quad \text{where} \quad h_0(x) = (x^{-1} + (3 h_1(x))^{-1})^{-1} = \frac{3 x h_1(x)}{x + 3 h_1(x)}, \quad x\in(0,\infty).
\]
Hence, the price impacts at the full participation equilibrium are
\[
	\Lambda^f_0 = V \tilde{h}_0(D^\rho)V,\qquad\Lambda^f_1 = V \tilde{h}_1(D^\rho)V,
\]
where 
\[
	\tilde{h}_0(x)=\frac{1}{3h_1(x)}  \quad \text{and} \quad \tilde{h}_1(x) = \frac{1}{h_0(x)+2h_1(x)} = \frac{x+3h_1(x)}{h_1(x)(5x+6h_1(x))}, \quad x\in(0,\infty).
\]

Prices at full-participation equilibrium are given by
\[
p^f = (X^f_0 + 3 X^f_1)^{-1} X^f_0 g_0 = V (h_0(D^\rho) + 3 h_1(D^\rho))^{-1} h_0(D^\rho) V g_0.
\]
We directly obtain from~\eqref{eq:price_impact_dem} that traders $i\in I_{-0}$ have equilibrium position $q^f_i = q^f_1$, where
\[
q^f_1 = - X^f_1 p^f = - V h_1(D^\rho) (h_0(D^\rho) + 3 h_1(D^\rho))^{-1} h_0(D^\rho) V g_0 = V \eta_1(D^\rho)  V g_0,
\] 
and
\[
\eta_1(x) = - h_1(x) (h_0(x) + 3 h_1(x))^{-1} h_0(x) = - \frac{h_1(x)}{1 + 3 h_1(x) / h_0(x)} = - \frac{x h_1(x)}{2 x + 3 h_1(x)}, \quad x\in(0,\infty).
\]

We then calculate that the \emph{aggregate} utility at Nash equilibrium in full participation equals
\[
\sum_{i\in I} \pare{\inner{q^f_i}{g_i} - \frac{1}{2 \delta_i}\inner{q^f_i}{C_i q^f_i}} =\inner{q^f_0}{g_0} - \frac{1}{2}\inner{q^f_0}{C_0 q^f_0} -3 \frac{1}{2}\inner{q^f_1}{ q^f_1} = g_0' V \kappa (D^\rho)  V g_0,
\]
where
\[
\kappa(x) = - 3 \eta_1(x) - \frac{9}{2}  \frac{\eta^2_1(x)}{x} -  \frac{3}{2} \eta^2_1(x), \quad x\in(0,\infty).
\]

\subsection{Proof of Lemma~\ref{lem:full_part_1}}\label{appsubsec:proof_lem_full_part_1}

Retain all notation of \S\ref{appsubsec:equilibrium_example_full_part}. Also, recall that
\[
	\frac{\gamma_1^2}{3} = \frac{1}{3} g'_0 \left(\begin{array}{c c}
		1 & 0 \\
		0 & 0
		\end{array} \right) g_0
\]
is the aggregate utility in restricted participation equilibrium; therefore, the difference between aggregate utilities in full and restricted participation equilibrium equals
\begin{equation}\label{eq:utility_difference}
g_0' V \kappa (D^\rho)  V g_0 - \frac{1}{3} g'_0 \left(\begin{array}{c c}
	1 & 0 \\
	0 & 0
	\end{array} \right) g_0 = g_0' V \left( \kappa (D^\rho) - \frac{1}{6} \left(\begin{array}{c c}
1 & 1 \\
1 & 1
\end{array} \right) \right) V g_0.
\end{equation}
To see whether this quadratic form may become negative, we check whether the smallest eigenvalue of the matrix
\[
M \equiv M^\rho :=\kappa (D^\rho) - \frac{1}{6} \left(\begin{array}{c c}
1 & 1 \\
1 & 1
\end{array} \right) =
\left(\begin{array}{c c}
	\kappa(1 / (1+\rho)) - 1/6& -1/6 \\
	-1/6 & \kappa(1 / (1-\rho)) - 1/6
	\end{array} \right)
\]
becomes negative. If this happens, we may choose $g_0$ so that $V g_0$ equals the eigenvector corresponding to this minimal eigenvalue, and the result of \eqref{eq:utility_difference} will be negative. The smallest of the two eigenvalues of the matrix $M$ is given by
\[
e_M(\rho):=\frac{\kappa((1+\rho)^{-1}) + \kappa((1-\rho)^{-1} ) -1/3 - \sqrt{(\kappa((1+\rho)^{-1}) - \kappa((1-\rho)^{-1}))^2 + 1/9} }{2}.
\]
In order to see analytically that this function may become negative, we analyse how $e_M(\rho)$ behaves when $\rho \approx 0$, which involves the behaviour of $\kappa(x)$ when $x \approx 1$. The function $\kappa$ depends on the function $\eta_1$, which in turn depends on the function $h_1$; all these functions involve algebraic expressions, and therefore derivatives can be explicitly. Straightforward (but lengthy and uninteresting) computations give $e_M(0) = 0$, $e_M'(0) = 0$, and $e_M''(0) = - 1/5000 < 0$. It follows that there exists $r \in (0, 1)$ such that, for $|\rho| < r$ and $\rho \neq 0$, we have $e_M(\rho) < 0$. In fact, and since $e_M$ is clearly an even function of $\rho \in (-1, 1)$, one may solve algebraically the equation $e_M(r) = 0$ for $r \in (0,1)$, obtaining a unique solution $r = (2/3) \sqrt{\sqrt{113} -9} \approx 0.851$. As a visual indication of the previous, Figure \ref{fig} presents a plot of $e_M(\rho)$ against $\rho \in (-1, 1)$.

\begin{center}
	\begin{figure}
	\includegraphics[scale=1]{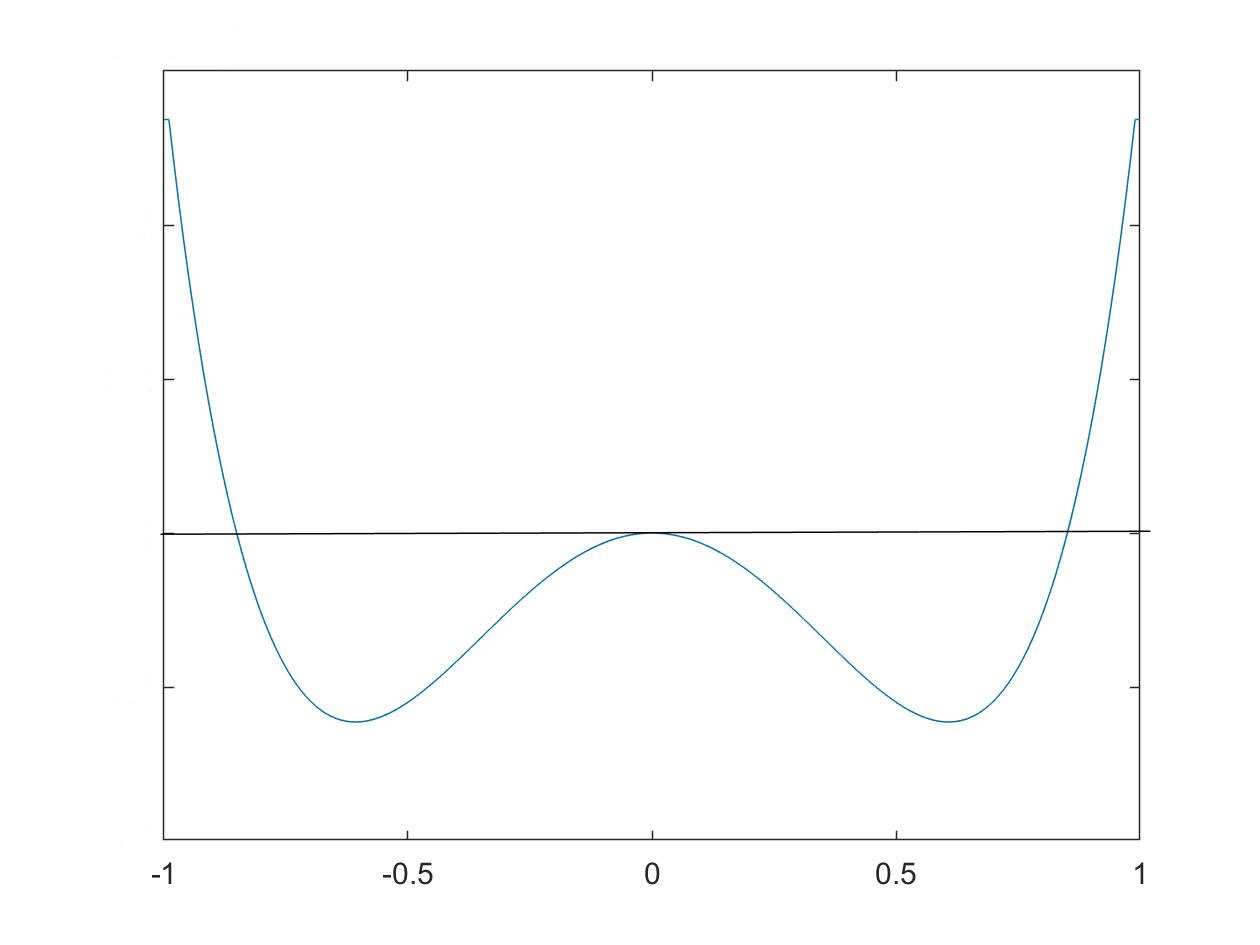}
	  \caption{The smallest eigenvalue of matrix $M$ as a function of trader 0's correlation parameter $\rho$. In the areas where this function is negative, there are initial risk exposures of trader 0 that yield higher welfare for the restricted market setting compared to the full participation.}\label{fig}
	\end{figure}
\end{center}

\subsection{Proof of Lemma~\ref{lem:full_part_2}}\label{appsubsec:proof_lem_full_part_2} 

Based on the discussion of \S\ref{appsubsec:equilibrium_example_full_part}, the inequality $\Lambda_0^f \preceq \Lambda_0^r$ is equivalent to $12 (V \Lambda_0^r V - \tilde{h}_0(D^\rho))$ being positive definite. For this, we compute
	\[
		12 V \Lambda^r_0 V = 6 \left(\begin{array}{c c}
			1 & 1 \\
			1 & -1
			\end{array} \right)
			\left(\begin{array}{c c}
	1/2 &0 \\
	0 &2/3
	\end{array} \right)  \left(\begin{array}{c c}
			1 & 1 \\
			1 & -1
			\end{array} \right) =
			\left(\begin{array}{c c}
			7 & -1 \\
			-1 & 7
			\end{array} \right).
	\]
	Furthermore, since
	\[
		12 \tilde{h}_0(x) = \frac{4}{1/4 - 5 x / 12 + \sqrt{ \left( 1/4 -  5 x / 12 \right)^2 + 2 x / 3}} = \frac{- 3/x + 5 + \sqrt{ \left( 3/x -  5 \right)^2 + 96 / x}}{2},
	\]
	with $\bar{h}_0 (x) = -5 - (3/2) x + \sqrt{ \left( 1 - (3/2) x \right)^2 +  24 (1 +x)}$ we have $6 + \bar{h}_0(x) = 12 \tilde{h}_0 (1 / (1+x))$, so that
	\[
		12 (V \Lambda^r_0 V - \tilde{h}_0 (D^\rho) ) = 
			\left(\begin{array}{c c}
			1 - \bar{h}_0(\rho)& -1 \\
			-1 & 1 - \bar{h}_0(-\rho)
			\end{array} \right).
	\]
	The above matrix will be positive definite if its trace is positive and its determinant nonnegative. It can be checked that $\bar{h}_0$ is concave, which implies that the trace $2 - \bar{h}_0(\rho) - \bar{h}_0(-\rho)$ is convex and (obviously) even in $\rho$; therefore it is always greater or equal than $2 - 2 \bar{h}_0 (0) = 2 > 0$. Furthermore, the determinant equals 
	\[
		(1 - \bar{h}_0(\rho))(1 - \bar{h}_0(-\rho)) - 1 = \bar{h}_0(\rho) \bar{h}_0(-\rho) - \bar{h}_0(\rho) - \bar{h}_0(-\rho);
	\]
	this function of $\rho$ is also even, and can be checked to be convex; therefore, it always dominates $(1 - \bar{h}_0(0))^2 - 1 = 0$.

\bibliographystyle{plainnat}
\bibliography{references}

\begin{thebibliography}{53}
\providecommand{\natexlab}[1]{#1}
\providecommand{\url}[1]{\texttt{#1}}
\expandafter\ifx\csname urlstyle\endcsname\relax
  \providecommand{\doi}[1]{doi: #1}\else
  \providecommand{\doi}{doi: \begingroup \urlstyle{rm}\Url}\fi

\bibitem[Almgren and Chriss(2000)]{AlmCh00}
R.~Almgren and N.~Chriss.
\newblock Optimal execution of portfolio transactions.
\newblock \emph{Journal of Risk}, 3:\penalty0 5--39, 2000.

\bibitem[Almgren et~al.(2005)Almgren, Thum, Hauptmann, and Li]{AlmThu05}
R.~Almgren, C.~Thum, E.~Hauptmann, and H.~Li.
\newblock Equity market impact.
\newblock \emph{Risk}, pages 57--62, July 2005.

\bibitem[Anthropelos(2017)]{Anth17}
M.~Anthropelos.
\newblock The effect of market power on risk-sharing.
\newblock \emph{Mathematics and Financial Economics}, 11:\penalty0 323--368,
  2017.

\bibitem[Anthropelos and Kardaras(2017)]{AnthKar17}
M.~Anthropelos and C.~Kardaras.
\newblock Equilibrium in risk-sharing games.
\newblock \emph{Finance and Stochastics}, 21\penalty0 (3):\penalty0 815--865,
  2017.

\bibitem[Anthropelos et~al.(2020)Anthropelos, Kardaras, and
  Vichos]{AnthrKardVich20}
M.~Anthropelos, C.~Kardaras, and G.~Vichos.
\newblock Effective risk aversion in thin risk-sharing markets.
\newblock \emph{Mathematical Finance}, 30:\penalty0 1565--1590, 2020.

\bibitem[Aouani and Cornet(2009)]{AouCor09}
Z.~Aouani and B.~Cornet.
\newblock Existence of financial equilibria with restricted participation.
\newblock \emph{Journal of Mathematical Economics}, 45\penalty0 (12):\penalty0
  772 -- 786, 2009.

\bibitem[Asker et~al.(2014)Asker, Farre-Mensa, and Ljungqvist]{AskFarLji14}
J.~Asker, J.~Farre-Mensa, and A.~Ljungqvist.
\newblock {Corporate Investment and Stock Market Listing: A Puzzle?}
\newblock \emph{The Review of Financial Studies}, 28\penalty0 (2):\penalty0
  342--390, 2014.

\bibitem[Babus and Partatore(2019)]{BabPart19}
A.~Babus and C.~Partatore.
\newblock Strategic fragmented markets.
\newblock Working paper, available at SSRN: https://ssrn.com/abstract=2856629,
  2019.

\bibitem[Bakshi et~al.(2015)Bakshi, Madan, and Panayotov]{Bak15}
G.~Bakshi, D.~Madan, and G.~Panayotov.
\newblock Heterogeneity in beliefs and volatility tail behavior.
\newblock \emph{Journal of Financial and Quantitative Analysis}, 50\penalty0
  (6):\penalty0 1389–1414, 2015.

\bibitem[Basak and Cuoco(1998)]{BasCuo98}
S.~Basak and D.~Cuoco.
\newblock An equilibrium model with restricted stock market participation.
\newblock \emph{The Review of Financial Studies}, 11\penalty0 (2):\penalty0
  309--341, 1998.

\bibitem[Biais et~al.(2005)Biais, Glosten, and Spatt]{BiaGloSpa05}
B.~Biais, L.~Glosten, and C.~Spatt.
\newblock Market microstructure: A survey of microfoundations, empirical
  results, and policy implications.
\newblock \emph{Journal of Financial Markets}, 8\penalty0 (2):\penalty0 217 --
  264, 2005.
\newblock ISSN 1386-4181.

\bibitem[Blake et~al.(2013)Blake, Timmermann, Tonks, and Wermers]{Bla13}
A.G. Blake, D.and~Rossi, A.~Timmermann, I.~Tonks, and R.~Wermers.
\newblock Decentralized investment management: Evidence from the pension fund
  industry.
\newblock \emph{The Journal of Finance}, 68\penalty0 (3):\penalty0 1133--1178,
  2013.

\bibitem[Bretscher et~al.(2020)Bretscher, Schmid, Sen, and
  Sharma]{BreSchSenSha20}
L.~Bretscher, L.~Schmid, I.~Sen, and V.~Sharma.
\newblock Institutional corporate bond pricing.
\newblock Swiss Finance Institute Research, 2020.

\bibitem[Busch(2017)]{Bus17}
D.~Busch.
\newblock 'agency and principal dealing under the markets in financial
  instruments directive (mifid) i and ii.
\newblock \emph{European Review of Private Law}, 25:\penalty0 337--362, 2017.

\bibitem[Calvet et~al.(2004)Calvet, Gonzalez-Eiras, and Sodini]{CavGonSod04}
L.~Calvet, M.~Gonzalez-Eiras, and P.~Sodini.
\newblock Financial innovation, market participation, and asset prices.
\newblock \emph{The Journal of Financial and Quantitative Analysis},
  39\penalty0 (3):\penalty0 431--459, 2004.

\bibitem[Carosi et~al.(2009)Carosi, Gori, and Villanacci]{CarGorVil09}
L.~Carosi, M.~Gori, and A.~Villanacci.
\newblock Endogenous restricted participation in general financial equilibrium.
\newblock \emph{Journal of Mathematical Economics}, 45\penalty0 (12):\penalty0
  787 -- 806, 2009.

\bibitem[Cass(2006)]{Cass06}
D.~Cass.
\newblock Competitive equilibrium with incomplete financial markets.
\newblock \emph{Journal of Mathematical Economics}, 42\penalty0 (4):\penalty0
  384--405, 2006.

\bibitem[Chen and Duffie(2021)]{CheDuff21}
D.~Chen and D.~Duffie.
\newblock Market fragmentation.
\newblock \emph{American Economic Review}, 111\penalty0 (7):\penalty0 2247--74,
  2021.

\bibitem[Cornet and Gopalan(2010)]{CorGop10}
B.~Cornet and R.~Gopalan.
\newblock {Arbitrage and equilibrium with portfolio constraints}.
\newblock \emph{Economic Theory}, 45\penalty0 (1):\penalty0 227--252, 2010.

\bibitem[Duchin and Levy(2010)]{Duc10}
R.~Duchin and M.~Levy.
\newblock Disagreement, portfolio optimization, and excess volatility.
\newblock \emph{Journal of Financial and Quantitative Analysis}, 45\penalty0
  (3):\penalty0 623–640, 2010.

\bibitem[Duffie(1987)]{Duff87}
D.~Duffie.
\newblock Stochastic equilibria with incomplete financial markets.
\newblock \emph{Journal of Economic Theory}, 41\penalty0 (2):\penalty0
  405--416, 1987.

\bibitem[Farmer and Joshi(2002)]{FarJoh02}
J.D. Farmer and S.~Joshi.
\newblock The price dynamics of common trading strategies.
\newblock \emph{Journal of Economic Behavior and Organization}, 49\penalty0
  (2):\penalty0 149 -- 171, 2002.

\bibitem[Frazzini et~al.(2018)Frazzini, Israel, and Moskowitz]{FraIsrMosk18}
A.~Frazzini, R.~Israel, and T.~Moskowitz.
\newblock Transaction costs.
\newblock Working paper, available at SSRN: https://ssrn.com/abstract=3229719,
  2018.

\bibitem[Fulkerson and Hong(2021)]{FulHon21}
J.~A. Fulkerson and X.~Hong.
\newblock Investment restrictions and fund performance.
\newblock \emph{Journal of Empirical Finance}, 64:\penalty0 317--336, 2021.

\bibitem[Hameed et~al.(2017)Hameed, Lof, and Suominen]{AllMatMat17}
A.~Hameed, M.~Lof, and M.~Suominen.
\newblock Slow trading and stock return predictability.
\newblock Working paper, available at SSRN: https://ssrn.com/abstract=2671237,
  2017.

\bibitem[Hau et~al.(2019)Hau, Hoffmann, Langfield, and Timmer]{HuHofLanTim19}
H.~Hau, P.~Hoffmann, S.~Langfield, and Y.~Timmer.
\newblock Discriminatory pricing of over-the-counter derivatives.
\newblock Working Paper, University of Geneva, 2019.

\bibitem[Hens et~al.(2006)Hens, Herings, and Predtetchinskii]{HenHerPre06}
T.~Hens, P.J-J. Herings, and A.~Predtetchinskii.
\newblock Limits to arbitrage when market participation is restricted.
\newblock \emph{Journal of Mathematical Economics}, 42\penalty0 (4):\penalty0
  556 -- 564, 2006.

\bibitem[Huberman and Werner(2004)]{HubSta04}
G.~Huberman and S.~Werner.
\newblock Price manipulation and quasi-arbitrage.
\newblock \emph{Econometrica}, 72\penalty0 (4):\penalty0 1247--1275, 2004.

\bibitem[Koijen and Yogo(2019)]{KoiYog19}
R.~S.~J. Koijen and M.~Yogo.
\newblock A demand system approach to asset pricing.
\newblock \emph{Journal of Political Economy}, 127\penalty0 (4):\penalty0
  1475--1515, 2019.

\bibitem[Kyle(1985)]{Kyl85}
A.S. Kyle.
\newblock Continuous auctions and insider trading.
\newblock \emph{Econometrica}, 53\penalty0 (6):\penalty0 1315--1335, 1985.

\bibitem[Kyle(1989)]{Kyl89}
A.S. Kyle.
\newblock Informed speculation with imperfect competition.
\newblock \emph{Review of Economic Studies}, 56\penalty0 (3):\penalty0 317--55,
  1989.

\bibitem[Malamud and Rostek(2017)]{MalRos17}
S.~Malamud and M.~Rostek.
\newblock Decentralized exchange.
\newblock \emph{American Economic Review}, 107\penalty0 (11):\penalty0
  3320--62, 2017.

\bibitem[Mas-Colell et~al.(1995)Mas-Colell, Whinston, and Green]{MasWhiGre}
A.~Mas-Colell, M.~Whinston, and J.~Green.
\newblock \emph{Microeconomic Theory}.
\newblock New York: Oxford University Press, 1995.

\bibitem[Neuhann and Sockin(2020)]{NeuSoc20}
D.~Neuhann and M.~Sockin.
\newblock Risk-sharing, investment, and asset prices according to cournot and
  arrow-debreu.
\newblock Working Paper, available at SSRN: https://ssrn.com/abstract=3320035
  or http://dx.doi.org/10.2139/ssrn.3320035, 2020.

\bibitem[OECD(2019)]{OECD19}
OECD.
\newblock
  \emph{https://www.oecd.org/daf/fin/private-pensions/2019-Survey-Investment-Regulation-Pension-Funds.pdf}.
\newblock 2019.

\bibitem[Pension and Investments(2017)]{PI17}
Pension and Investments.
\newblock \emph{{80\% of equity market cap held by institutions}}.
\newblock 2017.
\newblock Available online.

\bibitem[Polemarchakis and Siconolfi(1997)]{PolSic97}
H.M. Polemarchakis and P.~Siconolfi.
\newblock Generic existence of competitive equilibria with restricted
  participation.
\newblock \emph{Journal of Mathematical Economics}, 28\penalty0 (3):\penalty0
  289 -- 311, 1997.

\bibitem[Rahi and Zigrand(2009)]{RahZig09}
R.~Rahi and J.-P. Zigrand.
\newblock Strategic financial innovation in segmented markets.
\newblock \emph{Review of Financial Studies}, 22:\penalty0 2941--2971, 2009.

\bibitem[Rostek and Weretka(2015)]{RosWer15}
M.~Rostek and M.~Weretka.
\newblock {Dynamic thin markets}.
\newblock \emph{Review of Financial Studies}, 28:\penalty0 2946--2992, 2015.

\bibitem[Rostek and Weretka(2016)]{Rostek2016}
M.~Rostek and M.~Weretka.
\newblock Thin markets.
\newblock In \emph{The New Palgrave Dictionary of Economics}, pages 1--5.
  Palgrave Macmillan UK, 2016.

\bibitem[Rostek and Yoon(2020)]{RosYoo20}
M.~Rostek and J.~H. Yoon.
\newblock Equilibrium theory of financial markets: Recent developments.
\newblock Prepared for The Journal of Economic Literature, 2020.

\bibitem[Rostek and Yoon(2021)]{RosYoo21a}
M.~Rostek and J.~H. Yoon.
\newblock Exchange design and efficiency.
\newblock \emph{Econometrica}, 89\penalty0 (6):\penalty0 2887--2928, 2021.

\bibitem[Rostek and Yoon(2023)]{RosYoo23}
M.~Rostek and J.~H. Yoon.
\newblock Innovation in decentralized markets: Technology vs synthetic
  products.
\newblock \emph{American Economic Journal: Microeconomics}, Forthcoming, 2023.

\bibitem[{SEC, Office of Investor Education and Advocacy}(2016)]{SEC16}
{SEC, Office of Investor Education and Advocacy}.
\newblock \emph{Mutual Funds and ETFs: A Guide for Investors}.
\newblock SEC Pub. 182 (12/16)., 2016.

\bibitem[Tuckman and Vila(1992)]{TuckVi92}
B.~Tuckman and J.-L. Vila.
\newblock Arbitrage with holding costs: A utility-based approach.
\newblock \emph{The Journal of Finance}, 47\penalty0 (4):\penalty0 1283--1302,
  1992.

\bibitem[Vayanos(1999)]{Vay99}
D.~Vayanos.
\newblock Strategic trading and welfare in a dynamic market.
\newblock \emph{Review of Economic Studies}, 66\penalty0 (2):\penalty0 219--54,
  1999.

\bibitem[Vives(2008)]{Viv08}
X.~Vives.
\newblock \emph{Information and Learning in Markets: The Impact of Market
  Microstructure}.
\newblock Princeton University Press, 2008.

\bibitem[Vives(2011)]{Viv11}
X.~Vives.
\newblock Strategic supply function competition with private information.
\newblock \emph{Econometrica}, 79:\penalty0 1919--1966, 2011.

\bibitem[Weretka(2011)]{Wer11}
M.~Weretka.
\newblock Endogenous market power.
\newblock \emph{Journal of Economic Theory}, 146\penalty0 (6):\penalty0 2281 --
  2306, 2011.

\bibitem[Whitehead(2011)]{Whi11}
C.~K. Whitehead.
\newblock The volcker rule and evolving financial markets.
\newblock \emph{Harvard Business Law Review}, 1:\penalty0 11--19, 2011.

\bibitem[Wittwer(2021)]{Witt21}
M.~Wittwer.
\newblock Connecting disconnected financial markets?
\newblock \emph{American Economic Journal: Microeconomics}, 13\penalty0
  (1):\penalty0 252--82, 2021.

\bibitem[Zigrand(2004)]{Zig04}
J.-P. Zigrand.
\newblock A general equilibrium analysis of strategic arbitrage.
\newblock \emph{Journal of Mathematical Economics}, 40\penalty0 (8):\penalty0
  923 -- 952, 2004.

\bibitem[Zigrand(2006)]{Zig06}
J.-P. Zigrand.
\newblock Endogenous market integration, manipulation and limits to arbitrage.
\newblock \emph{Journal of Mathematical Economics}, 42\penalty0 (3):\penalty0
  301 -- 314, 2006.

\end{thebibliography}
\end{document}